\documentclass[reprint,longbibliography,
 amsmath,amssymb,
 prl,superscriptaddress]{revtex4-2}
\usepackage{graphicx}
\usepackage{dcolumn}
\usepackage{bm}
\usepackage[colorlinks=true, citecolor=blue, linkcolor=blue, urlcolor=blue]{hyperref}
\usepackage{xcolor}
\usepackage{amsmath, amssymb, amsfonts}
\usepackage{mathtools}
\usepackage{subfigure}
\usepackage{mathrsfs}
\usepackage{lipsum}
\usepackage{sidecap}
\usepackage{verbatim}
\usepackage{multirow}
\usepackage{amsthm}
\theoremstyle{definition}

\newtheorem{theorem}{Theorem}

\newtheorem{corollary}{Corollary}

\definecolor{purple1}{rgb}{128,0,128}
\definecolor{darkgreen1}{rgb}{0,0,128}

\usepackage{bigints}
\newcommand{\bea}{\begin{eqnarray}}
\newcommand{\ea}{\end{eqnarray}}

\definecolor{darkpastelgreen}{rgb}{0.01, 0.75, 0.24}
\makeatletter
\newcommand*{\rom}[1]{\expandafter\@slowromancap\romannumeral #1@}
\makeatother
\begin{document}
\title{\bm{Quantum Tsunamis in a Bose-Einstein Condensate: An Analogue of Nonlinear
Gravitational Wave Memory}}
\title{Emergence of small scale structure of effective space-time through the propagation of dispersive shocks wave in Bose-Einstein Condensates} 
\title{Quantum pressure: The cosmic censor of Bose-Einstein 
condensate shock waves}
\title{Singularity in Fluid dynamical model of analogue gravity }
\title{Occurrence of Penrose type singularity in 2+1 D analogue blackholes}
\title{Probing spacetime singularities in sonic black holes}
\title{Probing Penrose-type spacetime singularities in sonic black holes}
\title{Probing Penrose-type singularities in sonic black holes}
\author{Satadal Datta} 
\affiliation{Seoul National University, Department of Physics and Astronomy, Center for Theoretical Physics, Seoul 08826, Korea}
\affiliation{Institut Pprime, UPR 3346 CNRS–Universit\'e de Poitiers–ISAE ENSMA, France}
\author{Uwe R. Fischer}
\affiliation{Seoul National University, Department of Physics and Astronomy, Center for Theoretical Physics, Seoul 08826, Korea}
\date{\today}
\begin{abstract}
Addressing the general question whether Penrose singularities physically exist inside black holes, we investigate the problem in the context of an analogue system, a flowing laboratory liquid, for which the governing equations are at least in principle known to all relevant scales, and in all regions of the effective spacetime. We suggest to probe the physical phenomena taking place close to the singularity in the interior of a $2+1$D  analogue black hole arising from a polytropic, inviscid, irrotational, and axisymmetric steady flow, and propose to this end an experimental setup in a Bose-Einstein condensate.
Our study provides concrete evidence, for a well understood dynamical system, that the Einstein equations are not necessary for a singularity to form, demonstrating that Penrose-type spacetime singularities can potentially also exist in non-Einsteinian theories of gravity. 
Finally, we demonstrate how the singularity is physically avoided in our proposed laboratory setup.
\end{abstract}
\maketitle

Oppenheimer and Snyder \cite{Oppenheimer} provided the first explicit description of the formation of a Schwarzschild black hole \cite{Sb} from the spherically symmetric gravitational collapse of nonrotating cold neutron stardust, with a total mass exceeding the Tolman-Oppenheimer-Volkoff limit  ($M\sim 0.7 M_{\odot}$) \cite{Tolman,TOV}. 
The solution indicated the possible formation of spacetime singularity as the neutron star shrinks in size below the Schwarzschild radius, continuing to shrink until its radius becomes zero, so that mass and energy 
becoming locally infinite. The Penrose singularity theorem \cite{Penrose65PRL} falsifies 
Lifshitz and Khalatnikov's claim of no singularity in the Universe \cite{Lifshitzem}, for a more general process than a spherically symmetric collapse of dust. If a spacetime with a noncompact Cauchy hypersurface contains a trapped surface, and if the null energy condition is satisfied \cite{Witten}, then the spacetime is null geodesically incomplete: There then exists at least one null geodesic starting from the trapped surface that can not be extended to the future beyond a {\it caustic}, within a finite affine parameter interval. This characterizes a {\it Penrose singularity} \cite{Penrose65PRL, Hawking1972BlackHI, PenroseT, ECSS, Witten, Wald, ReviewPenrose}. 
A caustic (also called {\it focal point}) is a point in spacetime where the null expansion $\Theta$ 
\cite{Penrose65PRL, Witten, RC, Sachs} becomes negative infinity. During gravitational collapse (not necessarily spherically symmetric) in an asymptotically flat spacetime close to a Schwarzschild black hole  
the aforementioned conditions in the theorem can be satisfied, leading to a Penrose singularity. 

In what follows, we aim to establish the mathematical and  physical conditions for a Penrose-type singularity in a laboratory system. Our primary task is to clarify the elusive nature of spacetime 
singularities in a system which is at least in principle under complete microscopic control. 
Note, in particular, that under the cosmic censorship conjecture \cite{Penrose99,Wald99} the observational access to the singularity of a real black hole from the outside is prohibited by the existence of an event horizon.  It is 
crucial for our purpose of probing the physical nature of the singularity that, in the analogue scenario, the lab experimentalist {\it can} access the black hole interior. 

{Unruh has shown in 1981 \cite{unruh}  that for an inviscid barotropic fluid with zero vorticity, linearized  perturbations on top of a given background   
satisfy a minimally coupled massless scalar field equation in an effective spacetime with pseudo-Riemannian {\em acoustic} metric which depends on background flow density and velocity \cite{BLV}. If a time-independent background becomes supersonic in the flow's direction, sound remains trapped in the supersonic region, and the boundary to the subsonic region becomes 
an event horizon; thus an analogue black hole (ABH) is formed \cite{Unruh95}, opening the door to observe Hawking radiation experimentally. On the classical level, this led to 
the observations of stimulated Hawking emission in water tank setups 
\cite{Rousseaux2008,Weinfurtner,Euve,Unruh2014} and the quantum effect of spontaneous Hawking radiation
in a Bose-Einstein condensate (BEC) \cite{PhysRevLett.85.4643,barcelo2001analogue,ModPhysLett,Carusotto_2008, Steinhauer16,Munoz,Kolobov2021}. Other systems enabling Hawking radiation studies are, e.g., fiber optical systems \cite{Philbin1367,Rosenberg}, superfluid helium \cite{Volovik_1999,Volovik2006}, photon fluids  \cite{Marino,Nguyen,PhysRevB.86.144505}, and Weyl semimetals \cite{VolovikWSM_2016,WilczekWSM,VolovikWSM}.

We here probe a novel aspect of the spacetime of an ABH, a potential 
Penrose-type singularity at its center.
In our model setup, importantly, the bounded affine parameter value is associated to a finite lab time interval as well. For clarity, we should emphasize here that the Penrose singularity theorem is {\em not} a statement about the divergence of curvature scalars \cite{kerr2023blackholessingularities,Penrose}, 
a divergence frequently occurring in analogue models \cite{VolovikBriefReview, PRD2023}. 
We also note that while it has been demonstrated before that a curvature singularity 
can occur in {\em background-incompressible} quasi-one-dimensional models within a finite affine parameter distance \cite{Barcelo_2004}, a trapped surface 
cannot be contained in such a 1+1 D model which, while not required for Hawking radiation \cite{Sonego}, is essential for the Penrose singularity theorem to apply. Furthermore, the necessary 
ingredient 
Raychaudhuri equation \cite{RC} for a null geodesic congruence \cite{Sachs} is not even defined in $1+1$D \cite{Witten}. 
Since a trapped surface is a closed hypersurface, an ABH 
in $2+1$D or $3+1$D spacetime however would contain a trapped surface in the supersonic flow region. 
The non-negativity of local energy density required in the Penrose theorem 
is a consequence of Einstein's gravity. The acoustic spacetime metric is however here determined by nonrelativistic fluid dynamics, forming an example of a nonlinear scalar field theory 
of gravity \cite{Datta_2022} also cf., e.g. Ref.~\cite{Novello}.
We thus here demonstrate how a Penrose-type 
singularity could potentially occur in a non-Einsteinian theory of gravity as well.
By employing an ABH, we provide insight into the nature of the puzzle a 
Penrose-type singularity presents. In particular, our proposed experimental design allows us to verify how the singularity is {\it physically} avoided.
Below, we will provide two theorems to establish the mathematical occurrence of a 
Penrose-type singularity inside ABHs; details are provided in the Supplemental Material \cite{suppl}.

{\it{Axisymmetric two-dimensional steady flow.}} 
Stationary continuity  in cylindrical coordinates $(r,\phi,z)$ yields 
\begin{equation}\label{continuity}
\frac{d}{d r}(r\rho_{0} v_{0}^r)=0,
\end{equation}
where in general fluid velocity and density are $(v_{0}^r (r), v_{0}^\phi (r), 0)$ and  $\rho_{0}(r)$, respectively; 
thus we have the mass flow constant 
$r \rho_{0}v_{0}^r={\rm constant}=C_1$; the quantity $2\pi r \rho_{0}v_{0}^r H$ is a constant mass flow rate, with $H$ the vertical height of the system. 
Maintaining such a time-independent nonzero constant mass flow requires the presence of a source and a drain at some finite radius around $r=0$ \cite{Chikka, Caio}. 
Naturally, Eq.~\eqref{continuity} can not be extended to $r=0$ for physical flows with nonzero mass flow rate. Only limiting values of the fluid quantities as $r$ approaches zero are well defined. 
The origin, $r=0$ can be considered as the {\it fluid dynamical singularity} of an axisymmetric flow  with nonzero mass flow rate. We also note that $r=0$ is a {\it coordinate singularity} in cylindrical coordinates  ($r, \phi, z$). Deriving the Schwarzschild black hole solution from the spherically symmetric time-independent Einstein field equations, one is faced with a similar  singularity at the location of the point mass source.
The Euler equation in radial direction reads 
\begin{equation}\label{vr}
v_{0}^r\frac{d v_{0}^r}{d r}-\frac{(v_{0}^\phi)^2}{r}=-\frac{1}{\rho_{0}}\frac{d p_{0}}{d r}-\frac{d V_{\rm ext}}{d r}.
\end{equation}
We assume the external potential $V_{\rm ext}(r)$ to be a smooth function of $r$ in the domain of space of the finite sized system, 
a minimal requirement of our analysis. {Since no torque is applied to the fluid}, the Euler equation along $\phi$ gives conserved angular momentum, $r v_{0}^\phi\coloneqq l ={\rm constant}$. 


We assume a polytropic equation of state, $p=K\rho^\gamma$, where 
$\gamma > 1$ 
and $K>0$; the 
sound speed derives from 
$c_{s0}^2=\frac{dp}{d\rho}|_{p=p_{0},\rho=\rho_{0}}$. 
Integrating Eq. \eqref{vr}, we are left with Bernoulli's constant $\frac{1}{2} (v_{0}^r)^2 +\frac{l^2}{2r^2}+\frac{c_{s0}^2}{(\gamma -1)}+V_{\rm ext} (r)={\rm constant}\coloneqq C_2.$ 
The first-order ordinary differential equations (ODE) for axisymmetric steady flow with nonzero mass flow rate, Eq.~\eqref{continuity} and Eq.~\eqref{vr},  are well defined everywhere in the noncompact space $r>0$ if the latter region is a subset of the smooth real function $V_{\rm ext}(r)$'s domain. Additionally, $V_{\rm ext}(r)$ could be a smooth function of $r$ at $r=0$ as well. We call the region $r>0$ the {\it maximal domain $M$}, 
because in this region of space we define the first order ODEs, Eqs.~\eqref{continuity} and \eqref{vr}. 

Changing variables from $(\rho_{0}, v_{0})$ to $(c_{s0}, v_{0})$, we have 
\begin{eqnarray}\label{dcsdr}
\frac{d c_{s0}}{dr}&=&-\frac{c_{s0}(\gamma -1)}{2}\left(\frac{1}{v_{0}^r}\frac{dv_{0}^r}{dr}+\frac{1}{r}\right), \\
\label{dvdr}
\frac{dv_{0}^r}{dr}&=&\frac{\frac{{c_{s0}}^2}{r}-\frac{d V_{\rm ext}}{d r}+\frac{l^2}{r^3}}{v_{0}^r\left(1-\frac{c_{s0}^2}{(v_{0}^r)^2 }\right)}.
\end{eqnarray}
Eqs.~\eqref{dcsdr} and \eqref{dvdr} are insensitive to the sign of $ v^r_{0}$, the radial direction of flow. Conventionally, we choose positive $C_1$. 
To solve this system of first order ODEs, specifying initial conditions of the initial value problem (IVP) is necessary, i,e., 
$c_{s0}, v^r_{0}$ at a radius $r_0$ are required. 
Given an IVP within a domain of $r$, if the solution $\rho_{0}(r)$ is positive, continuous and finite valued everywhere, and the $v_{0}^r (r), v_{0}^\phi (r)$ are continuous and finite valued, then the solution is a {\it steady physical flow} in that domain. The maximal domain of a physical steady flow $M_P$ depends on the IVP. Evidently, $M_P \subseteq M$. The region of space occupied by an axisymmetric steady flow in a finite sized system 
is a subspace of $M_P$. Insertion of the expression of $\frac{dv_{0}^r}{dr}$ of Eq. \eqref{dvdr} into the Eq. \eqref{dcsdr} yields $\frac{d c_{s0}}{dr}\coloneqq f(c_{s0}, v^r_{0}, r)$, and $\frac{d v^r_{0}}{dr}\coloneqq g(c_{s0}, v^r_{0},r)$. 
This leads us to formulate \cite{suppl} 
\begin{theorem}\label{thm1}
Given an IVP at $r=r_0$ with finite $\rho_{0}(r_0)>0$, and finite $ v^r_{0}(r_0)$ and $v^\phi_{0}(r_0)$, the maximal domain in $r$ of an axisymmetric steady physical flow $M_P$, with nonzero mass flux rate, is the region where the functions $f(c_{s0}, v^r_{0}, r)$ and $g(c_{s0}, v^r_{0},r)$ exist and are Lipschitz continuous \cite{grant2014theory,teschlordinary} with respect to $v^r_{0}$, and to $c_{s0}$ when $V_{\rm ext}(r)$ is smooth in $M$. 
\end{theorem}
\begin{theorem}\label{thm2}
If an axisymmetric steady physical flow  with nonzero mass flow rate under a smooth $V_{\rm ext}(r)$ exists in $M$, then $\lim_{r\to 0}V_{\rm ext}(r) =-\infty$.
\end{theorem}
\begin{corollary}
If an axisymmetric physical steady flow  with nonzero mass flow rate under a smooth $V_{\rm ext}(r)$ exists in the domain $0<r<R$, then $\lim_{r\to 0}V_{\rm ext}(r) =-\infty$.
\end{corollary}
Theorem \ref{thm2} implies that for a $V_{\rm ext}(r)$ which is smooth and finite everywhere, a time-independent axisymmetric flow with nonzero mass flow rate fails to exist in $M$. 
However, $\lim_{r\to 0}V_{\rm ext}(r) =-\infty$ is not a sufficient condition 
for steady radial flow to exist in the domain $M$,  
because it does not exclude the possibility that at some finite $r$, ${dv^r_{0}}/{dr}$ could be diverging. Eq. \eqref{continuity} confirms that ${d\rho_{0}}/{dr}$ would then need to diverge as well to counterbalance ${dv^r_{0}}/{dr}$, making the functions  $\rho_{0}(r)$ and $v^r_{0}(r)$ nonanalytic. If the IVP is given at such an $r$, theorem \ref{thm1} implies that the solution would not be unique in its neighborhood. Since steady flow would then not exist in $M$, i.e., if our sink is of infinitesimal radius, the flow has to become time dependent in  $M$. Therefore,  not only a proper choice of $V_{\rm ext}(r)$ creates a steady flow with nonzero mass flow rate to exist in $M$, but we also need to specify the IVP properly. 

{\it Acoustic metric.}
We consider the $z$ direction to be frozen (quasi-2D trapping geometry). Then, according to \cite{BLV,Visser_2005}, the physical acoustic metric in an axisymmetric 
time-independent $2+1$D flow reads 
\begin{equation}\label{agmn}
g_{\mu\nu}=\left(\frac{\rho_{0}}{c_{s0}}\right)^2\begin{bmatrix}
 -\left(c_{s0}^{2}-(v_{0}^r)^2-\frac{l^2}{r^2}\right) & -v_{0}^{r} & -l\\
-v_{0}^{r}&1&0 \\
-l&0 &r^2
\end{bmatrix}.
\end{equation}
The metric components (also those of the inverse metric) are finite quantities for physical flows. Therefore, the metric does not have a coordinate singularity in $M_P$ in the (inertial) nonrelativistic laboratory frame.

The event horizon in $2+1$D is the hypersurface where $g^{rr}=0$ \cite{suppl};  therefore $(c_{s0}(r_H))^2=(v_{0}^{r}(r_H))^2$ (a static limit surface arises as well 
from nonvanishing angular momentum of the flow  \cite{suppl}). For a black hole event horizon, $v_{0}^{r}$ is negative ($v_0^r=-v_0$). 
Inside the ABH event horizon where the flow is supersonic, at constant $t$, the $r=$ constant closed hypersurface $S^1$ is a one-dimensional trapped surface $T^1$, i.e., null geodesics pointing outward orthogonal to the trapped surface also converge initially at $T^1$. We thus generalize the trapped surface concept in $3+1$D spacetimes to $2+1$D spacetimes; this is elaborated further in the SM \cite{suppl}.

{\it Analytical description of calculating the event horizon radius.} 
We map the coupled ODEs, Eq. \eqref{dcsdr} and Eq. \eqref{dvdr}
 to a {\em dynamical system} equations \cite{strogatz2018nonlinear}, by choosing a parameter $p$ such that 
 \begin{equation}
  \frac{d v_{0}^r}{dp}=\frac{{c_{s0}}^2}{r}-\frac{d V_{\rm ext}}{d r}+\frac{l^2}{r^3},
  \quad\,\, \frac{dr}{dp}=v_{0}^r\left(1-\frac{c_{s0}^2}{(v_{0}^r)^2 }\right); 
  \end{equation}
   and ${d c_{s0}}/{dp}$ then directly follows from Eq.~\eqref{dcsdr}. We can identify the event horizon condition by assuming finite ${dv_{0}^r }/{dr}$ at $r=r_H$, with fixed points conditions (${dr}/{dp}=0$, ${dv_{0}^r}/{dp}=0$), and then determine criteria for suitable steady-state solutions corresponding to real-valued horizon radii. 
 The corresponding method was first introduced in  \cite{Muchotrzeb-Czerny1986}, and has subsequently 
 been extensively used in the astrophysical literature cf., e.g., Refs.~\cite{a1, Abraham_2006,Bilic_2014, Datta2016}. 
For an external potential $V_{\rm ext}(r)=-\frac{V_0}{r}$ ($V_0>0$, satisfying the condition of  theorem \ref{thm2}) and for steady and purely radial ($l=0$) flow, we have an event horizon at $r_H=\frac{(3-\gamma)}{2(\gamma -1)C_2}V_0$. We work in units such that $V_0 =r_H =1$, which also implies
$c_{sH}=\sqrt{{V_0}/{r_H}}=1$.
Using the fixed point values as initial values, we numerically solve the IVP of Eqs.~\eqref{dcsdr} 
and \eqref{dvdr} \cite{suppl}.
\begin{figure}[t] 
\includegraphics[scale=0.58]{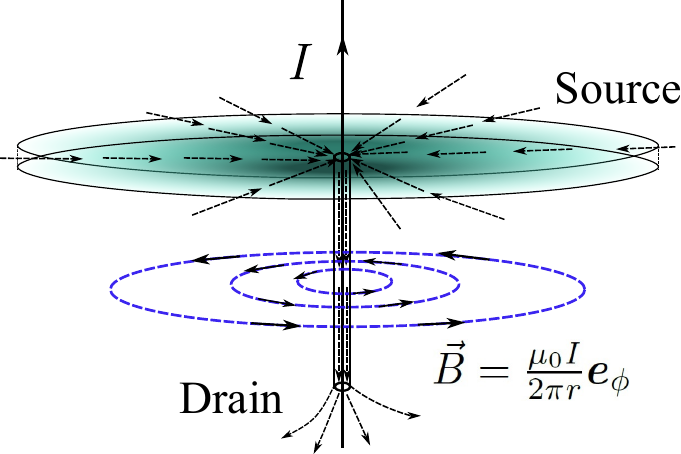}
\caption{\label{fsetup} Proposed setup for axisymmetric purely radial BEC flow creating an analogue black hole, 
BEC atoms are strongly confined in the $z$ direction (quasi-2D geometry). A steady current along $z$ creates an 
azimuthal magnetic field $\vec{B}$ according to Amp\`ere's law, 
leading to an external potential for atoms with magnetic moment flowing radially inward.
A coherent source at the boundary and a drain at the center 
(outcoupling the atoms \cite{Blochcw,laserO}) 
maintains a steady flow.}
\end{figure}

{\it Schematic of an experimental design.} {The potential $V_{\rm ext}\propto -\frac{1}{r}$ can 
be created in a BEC of atoms with magnetic moment $\tilde{\mu}>0$ 
oriented along an external magnetic field, cf.~Fig.~\ref{fsetup}.  
If we ignore the so-called quantum pressure (at sufficiently large wavelengths) and thus do not probe the analogue trans-Planckian regime, 
the Gross-Pitaevski\v\i\/ equation 
can be mapped to inviscid irrotational fluid dynamics, with equation of state, $p=\frac{1}{2} g\rho^2$, where $g>0$ is proportional to the $s$-wave scattering length in a dilute BEC with repulsive contact interactions 
\cite{RevModPhys.71.463}. In an external magnetic field $\vec{B}$, $V_{\rm ext}(r)=-\tilde{\mu} B$ \cite{pethick2002bose}. 
A steady current $I$ in a long wire along $z$ produces an external potential 
of the form 
$ V_{\rm ext}(r)=-\tilde{\mu} B=-\frac{\tilde{\mu}\mu _0 I}{2\pi r}$, 
where $\mu _0$ is the magnetic permeability of vacuum.

The {\it null expansion} $\Theta=l^\mu{}_{;\mu}$, which specifies 
the convergence or divergence of null geodesics with tangent $l^\mu$, 
is governed by the Raychaudhuri equation \cite{RC} as written for null geodesics \cite{RC,Hawking_Ellis_1973,dadhichderivation,Kar2007}. 
It is identical for both radially in- and outward pointing null geodesic congruences inside 
the trapped surface of our ABH \cite{suppl}. 
For radial rays in affine parametrization it reads 
\begin{equation}
\Theta = \frac{1}{2}\left(-\frac{v_0}{c_{s0}}+\frac{r}{c_{s0}}\frac{dv_0}{dr}\right). 
\end{equation}
It approaches minus infinity as $r\rightarrow 0$, identifying $r=0$ as a caustic point \cite{suppl}, 
cf.~Fig.~\ref{fig1}.

\begin{figure}[t]
\vspace*{0.5em}
\centering
\hspace*{-0.5em}
\includegraphics[scale=0.16]{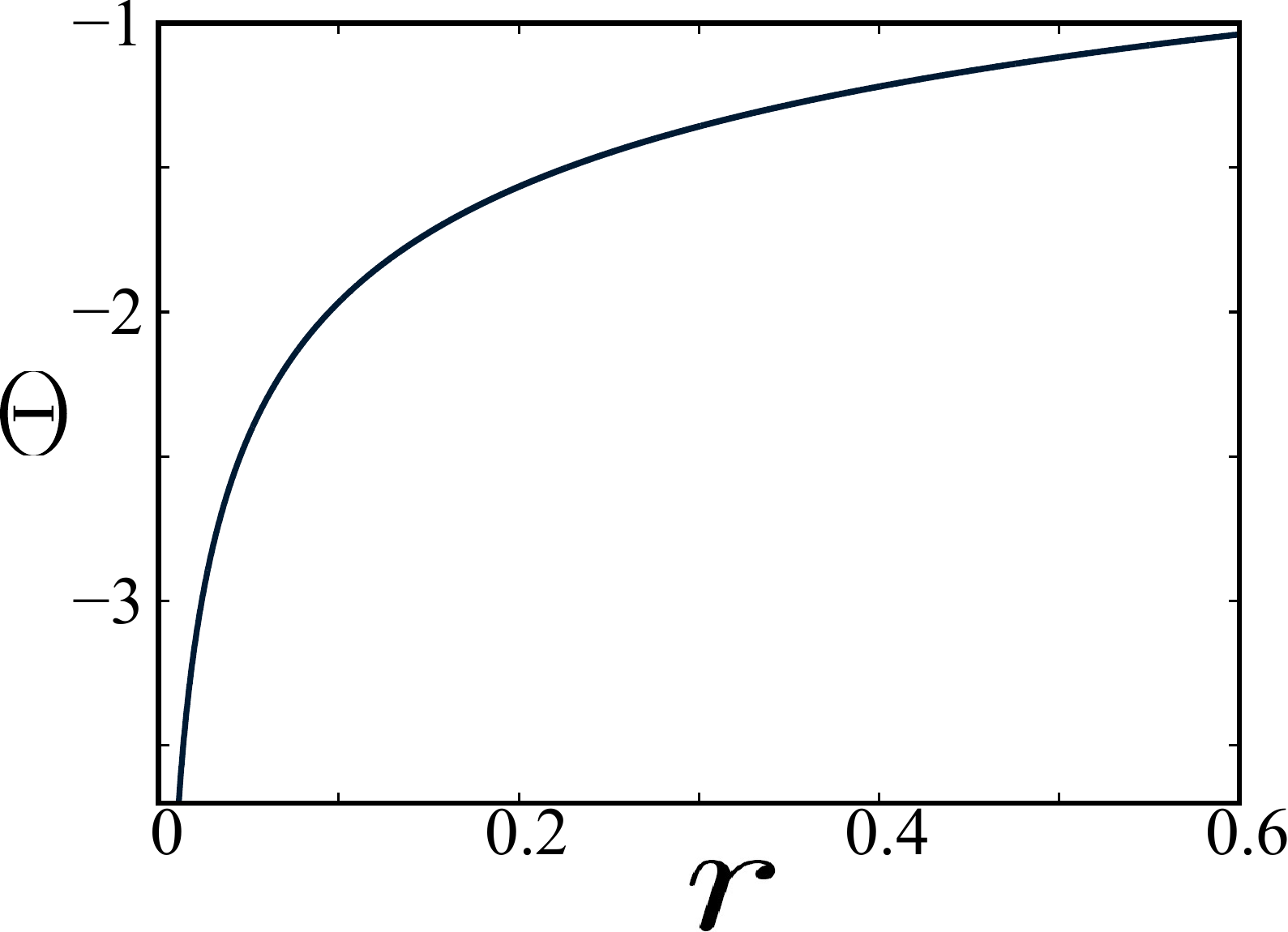}\includegraphics[scale=0.16]{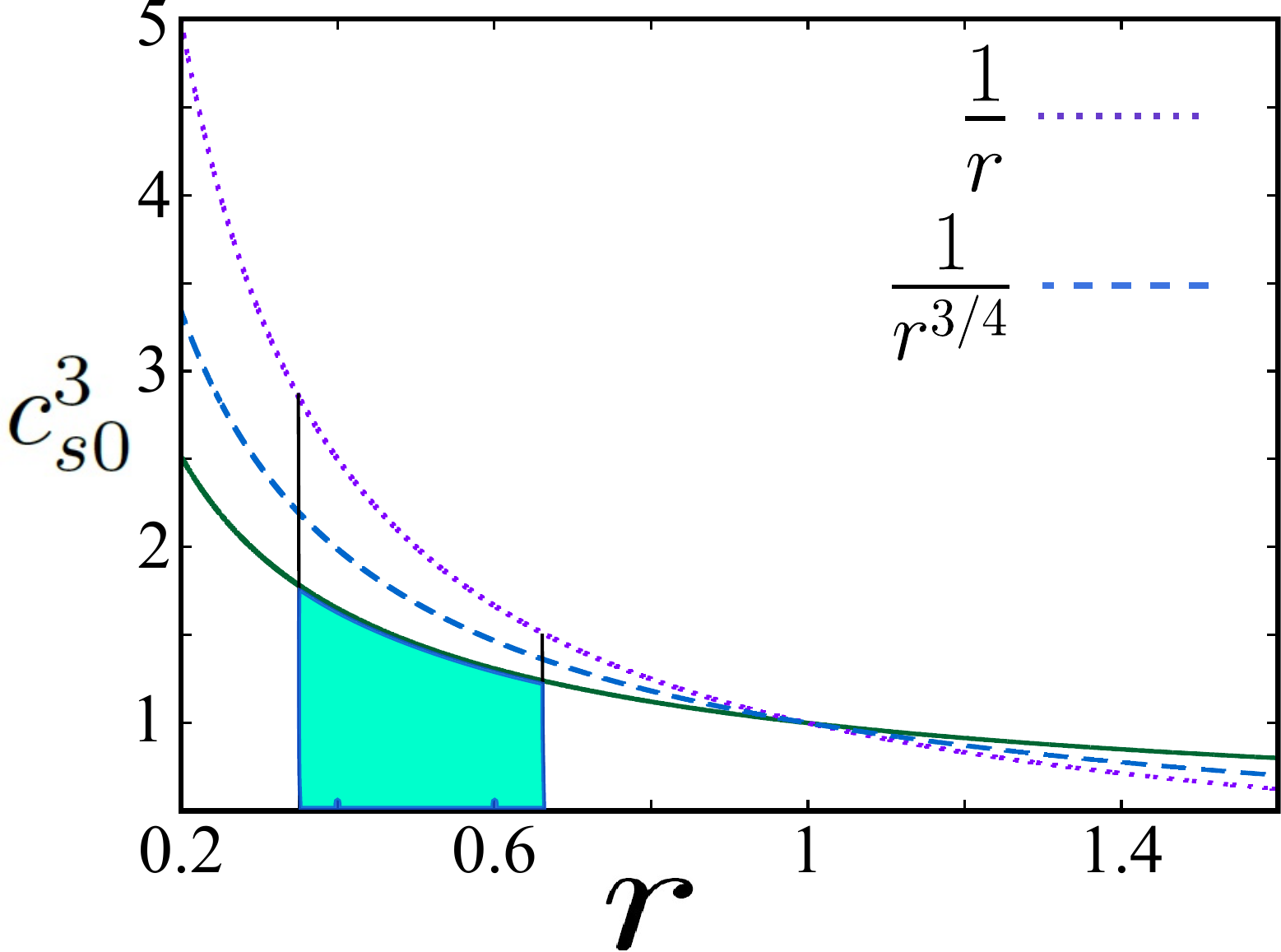} 
\caption{{\it Singularity in our ABH} from a transonic flow with $C_1=1$, $C_2=0.5$. 
(Left) The null expansion $\Theta$ starts from a negative value from the trapped surface at $r=0.6$ and it reaches minus infinity (see \cite{suppl}) when $r$ goes to zero. (Right) Affine parameter interval (see for further discussion in main text) between two given radii for radially in- and outward pointing 
null geodesics is the sea green area under the $c_{s0}^3 (r)$ curve (solid green curve).  
The affine interval is bounded for 
such geodesics traveling from the trapped surface to arbitrarily close to the origin, because $c_{s0}$ decays even
less fast than $r^{-1/4}$ in the supersonic flow region, cf.~Eq.~\eqref{affine}.}

\label{fig1}
\end{figure}

{\it Affine parameter interval.} For radially ingoing or outgoing null geodesic congruences in the acoustic metric \eqref{agmn} with $l=0$, the change in affine parameter $\lambda $ is given by 
\begin{equation} \label{affine}
d\lambda =\left(\frac{\rho _0}{c_{s0}}\right)^2 c_{s0} dr=c_{s0}^3 dr,
\end{equation}
 in units for which the BEC coupling $g=1$ \cite{suppl} .  
The sign in front of the affine parameter is immaterial due to affine transformation invariance 
{\cite{carroll_2019}; one chooses conventionally increasing affine parameter values for the future direction. We observe that the integral $\int _{r=r_0} ^ {r=0+}  (1/r^s) dr=\frac{1}{(-s+1)}\left[r^{-s+1}\right]_{r_0}^{0+}$ (for $s\neq 1$) is finite for real $s<1$,  therefore, the affine parameter interval traversed by a null geodesic traveling in the supersonic flow region to arbitrarily close to $r=0$ is bounded iff $c_{s0}< r^{-1/3}$ in the neighborhood of $r=0$, cf.~Fig.~\ref{fig1}. The presence of the closed future null boundary $B^2$, 
cf.~Fig.~\ref{penrose} confirms the null geodesic incompleteness \cite{Penrose65PRL, Witten} at $r=0$. Fluid-dynamically, it is impossible to have radially outgoing null geodesics in a supersonic flow region,
 and  null geodesics cannot be extended to the future beyond $r=0$. 
Hence this is not only a caustic point, but also {\it the Penrose-type singularity} of our ABH spacetime 
We also note that not only the affine interval for in- or outward pointing null geodesic to reach a point arbitrarily close to $r=0$ 
is bounded, but also that the corresponding lab time interval is bounded \cite{suppl}. 
However, the affine parameter is more suitable to describe 
null geodesics, especially in the supersonic flow region where the lab
 time axis is spacelike in terms of the acoustic metric. Sound pulses starting, e.g., from the trapped surface can be generated by suddenly turning on localized  repulsive laser potentials \cite{Andrews} 
 in our setup shown in Fig.~\ref{fsetup}.
 
Note that in our given case of spherical symmetry and radial null geodesics, 
it {\em coincidentally} happens that affine and lab time intervals are both finite. 
A truncation of null geodesics after a finite {\em lab time} interval then 
leads the experimentalist to unequivocally conclude that there must be a singularity in the sonic spacetime. This however must not hold for general flows and general null geodesics.

Using the numerical background flow solution of Eqs.~\eqref{dcsdr} 
and \eqref{dvdr} and using the relation $\frac{dr}{dt}=(-v_{0}\pm c_{s0})$ for radially out/inward pointing null geodesics, we produce Fig.~\ref{penrose}, which is akin to the one displayed by Penrose  
\cite{Penrose65PRL}. In distinction to the 3+1D case, our ABH is $2+1$D and does not possess a coordinate singularity at the horizon in the $(t,x,y)$ coordinates of the lab. The causal structure  of our ABH, displayed in Fig. \ref{penrose}, is 
in the lab coordinates similar to that of a Schwarzschild black hole in Eddington-Finkelstein coordinates \cite{Penrose65PRL,Eddington,Finkelstein}. 


\begin{figure}[t]
\vspace*{0.5em}
\centering
\includegraphics[scale=0.34]{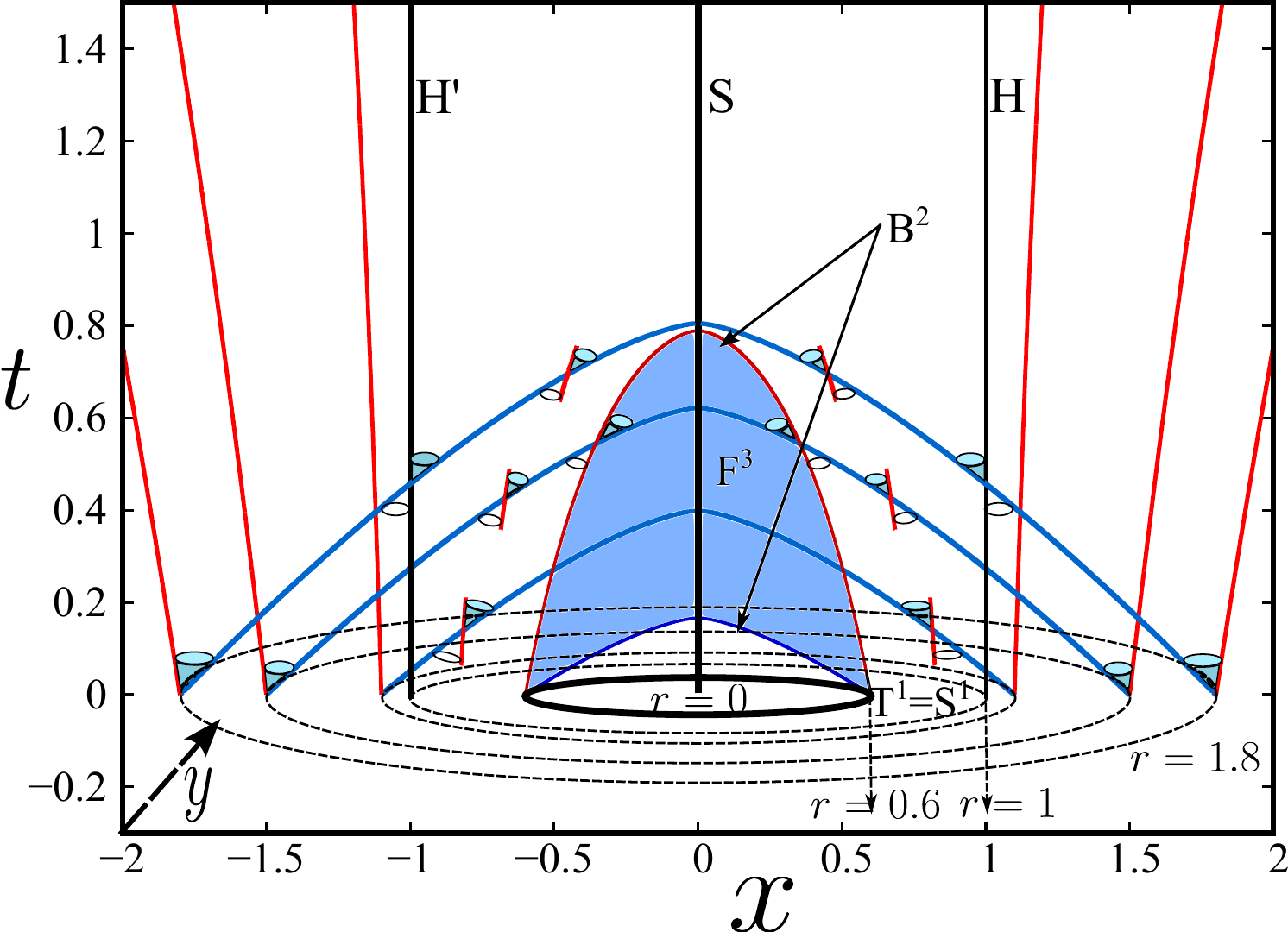}
\caption{{\it Causal structure of the ABH akin to the Penrose construction \cite{Penrose65PRL}}. 
Parameters as in Fig.~\ref{fig1}. Outward and inward pointing null geodesics in red and blue, respectively. The horizon, denoted by $H-H'$, forms a cylindrical hypersurface (in our units located at $r=1$),  
and the singularity ($S$) is at $r=0$. 
At $t=0$, a trapped surface $T^1$ is a 
one-dimensional circle $S^1$, here of radius 0.6, and the three dimensional $F^3$ is the future region occupied by smooth timelike curves from $T^1$, with the two-dimensional null boundary $B^2$, a closed hypersurface. 
Atoms escape from the $x-y$ plane through the drain around $r=0$, cf. Fig.~\ref{fsetup}. 
}
\label{penrose}
\end{figure}
{\it Discussion.}  
We have provided the first concrete laboratory analogue spacetime for 
which a Penrose-type singularity occurs 
at the level of Eulerian hydrodynamics. This is akin to the case of a Schwarzschild black hole singularity
obtained from Einstein gravity.  
The fluid-dynamical equations 
break down very close to the singularity,  
and a stable physical flow requires a drain of finite radius so the domain $M_P$ 
becomes well defined (cf.~Theorem \ref{thm1}); atoms escaping via an extra dimension 
(which is accessible to laboratory observers) thus here resolves the singularity. 
We envisage that singularity theorems inspired by our ABH example can also 
be formulated for more general non-Einsteinian theories of gravity \cite{suppl}. 

For the BEC ABH, significant density gradients  
near the singularity develop, 
and quantum pressure kicks in.
Dispersive corrections then modify the acoustic metric 
 \cite{PRD2023};  
for further increasing gradients, the very concept of a {\em local} acoustic metric 
can even break down completely, leading to nonlocal so-called ``rainbow" spacetimes \cite{rainbow,rainbowWV}. We are then entering an effective trans-Planckian regime already at the semiclassical (mean-field) level, We expect such a regime to occur for real gravity as well, 
either at the semiclassical or full quantum level, the Einstein equations being only valid in a similar hydro- and thermodynamical long-wavelength limit \cite{Ted}.

To address the question of the very existence of singularities in the absence of 
an observationally verified theory of quantum gravity, Kerr \cite{kerr2023blackholessingularities} 
recently has argued that even remaining completely within general relativity, it is debatable whether real black holes forming from an astrophysical process contain singularities.  
By providing a solvable laboratory ABH model in a fluid, 
we have indeed 
shown how its singularity is physically avoided 
already on a classical level. 




{\it Acknowledgments.}  
This work has been supported by the National Research Foundation of Korea under 
Grants No.~2017R1A2A2A05001422 and No.~2020R1A2C2008103. SD acknowledges support 
by the CNRS Grant No.~ANR-22-CPJ2-0039-01.

\bibliography{penrose_singularity_v24}

\begin{thebibliography}{82}%
\makeatletter
\providecommand \@ifxundefined [1]{%
 \@ifx{#1\undefined}
}%
\providecommand \@ifnum [1]{%
 \ifnum #1\expandafter \@firstoftwo
 \else \expandafter \@secondoftwo
 \fi
}%
\providecommand \@ifx [1]{%
 \ifx #1\expandafter \@firstoftwo
 \else \expandafter \@secondoftwo
 \fi
}%
\providecommand \natexlab [1]{#1}%
\providecommand \enquote  [1]{``#1''}%
\providecommand \bibnamefont  [1]{#1}%
\providecommand \bibfnamefont [1]{#1}%
\providecommand \citenamefont [1]{#1}%
\providecommand \href@noop [0]{\@secondoftwo}%
\providecommand \href [0]{\begingroup \@sanitize@url \@href}%
\providecommand \@href[1]{\@@startlink{#1}\@@href}%
\providecommand \@@href[1]{\endgroup#1\@@endlink}%
\providecommand \@sanitize@url [0]{\catcode `\\12\catcode `\$12\catcode
  `\&12\catcode `\#12\catcode `\^12\catcode `\_12\catcode `\%12\relax}%
\providecommand \@@startlink[1]{}%
\providecommand \@@endlink[0]{}%
\providecommand \url  [0]{\begingroup\@sanitize@url \@url }%
\providecommand \@url [1]{\endgroup\@href {#1}{\urlprefix }}%
\providecommand \urlprefix  [0]{URL }%
\providecommand \Eprint [0]{\href }%
\providecommand \doibase [0]{https://doi.org/}%
\providecommand \selectlanguage [0]{\@gobble}%
\providecommand \bibinfo  [0]{\@secondoftwo}%
\providecommand \bibfield  [0]{\@secondoftwo}%
\providecommand \translation [1]{[#1]}%
\providecommand \BibitemOpen [0]{}%
\providecommand \bibitemStop [0]{}%
\providecommand \bibitemNoStop [0]{.\EOS\space}%
\providecommand \EOS [0]{\spacefactor3000\relax}%
\providecommand \BibitemShut  [1]{\csname bibitem#1\endcsname}%
\let\auto@bib@innerbib\@empty
\bibitem [{\citenamefont {Oppenheimer}\ and\ \citenamefont
  {Snyder}(1939)}]{Oppenheimer}%
  \BibitemOpen
  \bibfield  {author} {\bibinfo {author} {\bibfnamefont {J.~R.}\ \bibnamefont
  {Oppenheimer}}\ and\ \bibinfo {author} {\bibfnamefont {H.}~\bibnamefont
  {Snyder}},\ }\bibfield  {title} {\bibinfo {title} {{On Continued
  Gravitational Contraction}},\ }\href {https://doi.org/10.1103/PhysRev.56.455}
  {\bibfield  {journal} {\bibinfo  {journal} {Phys. Rev.}\ }\textbf {\bibinfo
  {volume} {56}},\ \bibinfo {pages} {455} (\bibinfo {year} {1939})}\BibitemShut
  {NoStop}%
\bibitem [{\citenamefont {{Schwarzschild}}(1916)}]{Sb}%
  \BibitemOpen
  \bibfield  {author} {\bibinfo {author} {\bibfnamefont {K.}~\bibnamefont
  {{Schwarzschild}}},\ }\href@noop {} {\emph {\bibinfo {title} {{{\"U}ber das
  Gravitationsfeld eines Massenpunktes nach der Einsteinschen Theorie}}}}\
  (\bibinfo {year} {1916})\ pp.\ \bibinfo {pages} {189--196},\ \bibinfo {note}
  {{Sitzungsberichte der K\"oniglich Preussischen Akademie der
  Wissenschaften}}\BibitemShut {NoStop}%
\bibitem [{\citenamefont {Tolman}(1939)}]{Tolman}%
  \BibitemOpen
  \bibfield  {author} {\bibinfo {author} {\bibfnamefont {R.~C.}\ \bibnamefont
  {Tolman}},\ }\bibfield  {title} {\bibinfo {title} {{Static Solutions of
  Einstein's Field Equations for Spheres of Fluid}},\ }\href
  {https://doi.org/10.1103/PhysRev.55.364} {\bibfield  {journal} {\bibinfo
  {journal} {Phys. Rev.}\ }\textbf {\bibinfo {volume} {55}},\ \bibinfo {pages}
  {364} (\bibinfo {year} {1939})}\BibitemShut {NoStop}%
\bibitem [{\citenamefont {Oppenheimer}\ and\ \citenamefont
  {Volkoff}(1939)}]{TOV}%
  \BibitemOpen
  \bibfield  {author} {\bibinfo {author} {\bibfnamefont {J.~R.}\ \bibnamefont
  {Oppenheimer}}\ and\ \bibinfo {author} {\bibfnamefont {G.~M.}\ \bibnamefont
  {Volkoff}},\ }\bibfield  {title} {\bibinfo {title} {{On Massive Neutron
  Cores}},\ }\href {https://doi.org/10.1103/PhysRev.55.374} {\bibfield
  {journal} {\bibinfo  {journal} {Phys. Rev.}\ }\textbf {\bibinfo {volume}
  {55}},\ \bibinfo {pages} {374} (\bibinfo {year} {1939})}\BibitemShut
  {NoStop}%
\bibitem [{\citenamefont {Penrose}(1965)}]{Penrose65PRL}%
  \BibitemOpen
  \bibfield  {author} {\bibinfo {author} {\bibfnamefont {R.}~\bibnamefont
  {Penrose}},\ }\bibfield  {title} {\bibinfo {title} {{Gravitational Collapse
  and Space-Time Singularities}},\ }\href
  {https://doi.org/10.1103/PhysRevLett.14.57} {\bibfield  {journal} {\bibinfo
  {journal} {Phys. Rev. Lett.}\ }\textbf {\bibinfo {volume} {14}},\ \bibinfo
  {pages} {57} (\bibinfo {year} {1965})}\BibitemShut {NoStop}%
\bibitem [{\citenamefont {Lifshitz}\ and\ \citenamefont
  {Khalatnikov}(1963)}]{Lifshitzem}%
  \BibitemOpen
  \bibfield  {author} {\bibinfo {author} {\bibfnamefont {E.~M.}\ \bibnamefont
  {Lifshitz}}\ and\ \bibinfo {author} {\bibfnamefont {I.~M.}\ \bibnamefont
  {Khalatnikov}},\ }\bibfield  {title} {\bibinfo {title} {{Investigations in
  relativistic cosmology}},\ }\href {https://doi.org/10.1080/00018736300101283}
  {\bibfield  {journal} {\bibinfo  {journal} {Advances in Physics}\ }\textbf
  {\bibinfo {volume} {12}},\ \bibinfo {pages} {185} (\bibinfo {year}
  {1963})}\BibitemShut {NoStop}%
\bibitem [{\citenamefont {Witten}(2020)}]{Witten}%
  \BibitemOpen
  \bibfield  {author} {\bibinfo {author} {\bibfnamefont {E.}~\bibnamefont
  {Witten}},\ }\bibfield  {title} {\bibinfo {title} {{Light rays,
  singularities, and all that}},\ }\href
  {https://doi.org/10.1103/RevModPhys.92.045004} {\bibfield  {journal}
  {\bibinfo  {journal} {Rev. Mod. Phys.}\ }\textbf {\bibinfo {volume} {92}},\
  \bibinfo {pages} {045004} (\bibinfo {year} {2020})}\BibitemShut {NoStop}%
\bibitem [{\citenamefont {Hawking}(1972)}]{Hawking1972BlackHI}%
  \BibitemOpen
  \bibfield  {author} {\bibinfo {author} {\bibfnamefont {S.~W.}\ \bibnamefont
  {Hawking}},\ }\bibfield  {title} {\bibinfo {title} {{Black holes in general
  relativity}},\ }\href {https://api.semanticscholar.org/CorpusID:121527613}
  {\bibfield  {journal} {\bibinfo  {journal} {Communications in Mathematical
  Physics}\ }\textbf {\bibinfo {volume} {25}},\ \bibinfo {pages} {152}
  (\bibinfo {year} {1972})}\BibitemShut {NoStop}%
\bibitem [{\citenamefont {Penrose}(1972)}]{PenroseT}%
  \BibitemOpen
  \bibfield  {author} {\bibinfo {author} {\bibfnamefont {R.}~\bibnamefont
  {Penrose}},\ }\href {https://doi.org/10.1137/1.9781611970609} {\emph
  {\bibinfo {title} {{Techniques of Differential Topology in Relativity}}}}\
  (\bibinfo  {publisher} {Society for Industrial and Applied Mathematics},\
  \bibinfo {year} {1972})\BibitemShut {NoStop}%
\bibitem [{\citenamefont {Tipler}(1978)}]{ECSS}%
  \BibitemOpen
  \bibfield  {author} {\bibinfo {author} {\bibfnamefont {F.~J.}\ \bibnamefont
  {Tipler}},\ }\bibfield  {title} {\bibinfo {title} {{Energy conditions and
  spacetime singularities}},\ }\href {https://doi.org/10.1103/PhysRevD.17.2521}
  {\bibfield  {journal} {\bibinfo  {journal} {Phys. Rev. D}\ }\textbf {\bibinfo
  {volume} {17}},\ \bibinfo {pages} {2521} (\bibinfo {year}
  {1978})}\BibitemShut {NoStop}%
\bibitem [{\citenamefont {Wald}(2010)}]{Wald}%
  \BibitemOpen
  \bibfield  {author} {\bibinfo {author} {\bibfnamefont {R.~M.}\ \bibnamefont
  {Wald}},\ }\href {https://books.google.fr/books?id=9S-hzg6-moYC} {\emph
  {\bibinfo {title} {{General Relativity}}}}\ (\bibinfo  {publisher}
  {University of Chicago Press},\ \bibinfo {year} {2010})\BibitemShut {NoStop}%
\bibitem [{\citenamefont {Senovilla}\ and\ \citenamefont
  {Garfinkle}(2015)}]{ReviewPenrose}%
  \BibitemOpen
  \bibfield  {author} {\bibinfo {author} {\bibfnamefont {J.~M.~M.}\
  \bibnamefont {Senovilla}}\ and\ \bibinfo {author} {\bibfnamefont
  {D.}~\bibnamefont {Garfinkle}},\ }\bibfield  {title} {\bibinfo {title} {{The
  1965 Penrose singularity theorem}},\ }\href
  {https://doi.org/10.1088/0264-9381/32/12/124008} {\bibfield  {journal}
  {\bibinfo  {journal} {Classical and Quantum Gravity}\ }\textbf {\bibinfo
  {volume} {32}},\ \bibinfo {pages} {124008} (\bibinfo {year}
  {2015})}\BibitemShut {NoStop}%
\bibitem [{\citenamefont {Raychaudhuri}(1955)}]{RC}%
  \BibitemOpen
  \bibfield  {author} {\bibinfo {author} {\bibfnamefont {A.~K.}\ \bibnamefont
  {Raychaudhuri}},\ }\bibfield  {title} {\bibinfo {title} {{Relativistic
  Cosmology. I}},\ }\href {https://doi.org/10.1103/PhysRev.98.1123} {\bibfield
  {journal} {\bibinfo  {journal} {Phys. Rev.}\ }\textbf {\bibinfo {volume}
  {98}},\ \bibinfo {pages} {1123} (\bibinfo {year} {1955})}\BibitemShut
  {NoStop}%
\bibitem [{\citenamefont {Sachs}\ and\ \citenamefont {Bondi}(1961)}]{Sachs}%
  \BibitemOpen
  \bibfield  {author} {\bibinfo {author} {\bibfnamefont {R.}~\bibnamefont
  {Sachs}}\ and\ \bibinfo {author} {\bibfnamefont {H.}~\bibnamefont {Bondi}},\
  }\bibfield  {title} {\bibinfo {title} {{Gravitational waves in general
  relativity. VI. The outgoing radiation condition}},\ }\href
  {https://doi.org/10.1098/rspa.1961.0202} {\bibfield  {journal} {\bibinfo
  {journal} {Proceedings of the Royal Society of London. Series A. Mathematical
  and Physical Sciences}\ }\textbf {\bibinfo {volume} {264}},\ \bibinfo {pages}
  {309} (\bibinfo {year} {1961})}\BibitemShut {NoStop}%
\bibitem [{\citenamefont {Penrose}(1999)}]{Penrose99}%
  \BibitemOpen
  \bibfield  {author} {\bibinfo {author} {\bibfnamefont {R.}~\bibnamefont
  {Penrose}},\ }\bibfield  {title} {\bibinfo {title} {The question of cosmic
  censorship},\ }\href {https://doi.org/10.1007/BF02702355} {\bibfield
  {journal} {\bibinfo  {journal} {Journal of Astrophysics and Astronomy}\
  }\textbf {\bibinfo {volume} {20}},\ \bibinfo {pages} {233 } (\bibinfo {year}
  {1999})}\BibitemShut {NoStop}%
\bibitem [{\citenamefont {Wald}(1999)}]{Wald99}%
  \BibitemOpen
  \bibfield  {author} {\bibinfo {author} {\bibfnamefont {R.~M.}\ \bibnamefont
  {Wald}},\ }\bibfield  {title} {\bibinfo {title} {{Gravitational collapse and
  cosmic censorship}},\ }in\ \href@noop {} {\emph {\bibinfo {booktitle} {Black
  holes, gravitational radiation and the universe}}}\ (\bibinfo  {publisher}
  {Springer},\ \bibinfo {year} {1999})\ pp.\ \bibinfo {pages}
  {69--86}\BibitemShut {NoStop}%
\bibitem [{\citenamefont {Unruh}(1981)}]{unruh}%
  \BibitemOpen
  \bibfield  {author} {\bibinfo {author} {\bibfnamefont {W.~G.}\ \bibnamefont
  {Unruh}},\ }\bibfield  {title} {\bibinfo {title} {{Experimental Black-Hole
  Evaporation?}},\ }\href {https://doi.org/10.1103/PhysRevLett.46.1351}
  {\bibfield  {journal} {\bibinfo  {journal} {Phys. Rev. Lett.}\ }\textbf
  {\bibinfo {volume} {46}},\ \bibinfo {pages} {1351} (\bibinfo {year}
  {1981})}\BibitemShut {NoStop}%
\bibitem [{\citenamefont {Barcel{\'o}}\ \emph {et~al.}(2011)\citenamefont
  {Barcel{\'o}}, \citenamefont {Liberati},\ and\ \citenamefont {Visser}}]{BLV}%
  \BibitemOpen
  \bibfield  {author} {\bibinfo {author} {\bibfnamefont {C.}~\bibnamefont
  {Barcel{\'o}}}, \bibinfo {author} {\bibfnamefont {S.}~\bibnamefont
  {Liberati}},\ and\ \bibinfo {author} {\bibfnamefont {M.}~\bibnamefont
  {Visser}},\ }\bibfield  {title} {\bibinfo {title} {{Analogue Gravity}},\
  }\href {https://doi.org/10.12942/lrr-2011-3} {\bibfield  {journal} {\bibinfo
  {journal} {Living Reviews in Relativity}\ }\textbf {\bibinfo {volume} {14}},\
  \bibinfo {pages} {3} (\bibinfo {year} {2011})}\BibitemShut {NoStop}%
\bibitem [{\citenamefont {Unruh}(1995)}]{Unruh95}%
  \BibitemOpen
  \bibfield  {author} {\bibinfo {author} {\bibfnamefont {W.~G.}\ \bibnamefont
  {Unruh}},\ }\bibfield  {title} {\bibinfo {title} {{Sonic analogue of black
  holes and the effects of high frequencies on black hole evaporation}},\
  }\href {https://doi.org/10.1103/PhysRevD.51.2827} {\bibfield  {journal}
  {\bibinfo  {journal} {Phys. Rev. D}\ }\textbf {\bibinfo {volume} {51}},\
  \bibinfo {pages} {2827} (\bibinfo {year} {1995})}\BibitemShut {NoStop}%
\bibitem [{\citenamefont {Rousseaux}\ \emph {et~al.}(2008)\citenamefont
  {Rousseaux}, \citenamefont {Mathis}, \citenamefont {Maïssa}, \citenamefont
  {Philbin},\ and\ \citenamefont {Leonhardt}}]{Rousseaux2008}%
  \BibitemOpen
  \bibfield  {author} {\bibinfo {author} {\bibfnamefont {G.}~\bibnamefont
  {Rousseaux}}, \bibinfo {author} {\bibfnamefont {C.}~\bibnamefont {Mathis}},
  \bibinfo {author} {\bibfnamefont {P.}~\bibnamefont {Maïssa}}, \bibinfo
  {author} {\bibfnamefont {T.~G.}\ \bibnamefont {Philbin}},\ and\ \bibinfo
  {author} {\bibfnamefont {U.}~\bibnamefont {Leonhardt}},\ }\bibfield  {title}
  {\bibinfo {title} {{Observation of negative-frequency waves in a water tank:
  a classical analogue to the Hawking effect?}},\ }\href
  {https://doi.org/10.1088/1367-2630/10/5/053015} {\bibfield  {journal}
  {\bibinfo  {journal} {New Journal of Physics}\ }\textbf {\bibinfo {volume}
  {10}},\ \bibinfo {pages} {053015} (\bibinfo {year} {2008})}\BibitemShut
  {NoStop}%
\bibitem [{\citenamefont {Weinfurtner}\ \emph {et~al.}(2011)\citenamefont
  {Weinfurtner}, \citenamefont {Tedford}, \citenamefont {Penrice},
  \citenamefont {Unruh},\ and\ \citenamefont {Lawrence}}]{Weinfurtner}%
  \BibitemOpen
  \bibfield  {author} {\bibinfo {author} {\bibfnamefont {S.}~\bibnamefont
  {Weinfurtner}}, \bibinfo {author} {\bibfnamefont {E.~W.}\ \bibnamefont
  {Tedford}}, \bibinfo {author} {\bibfnamefont {M.~C.~J.}\ \bibnamefont
  {Penrice}}, \bibinfo {author} {\bibfnamefont {W.~G.}\ \bibnamefont {Unruh}},\
  and\ \bibinfo {author} {\bibfnamefont {G.~A.}\ \bibnamefont {Lawrence}},\
  }\bibfield  {title} {\bibinfo {title} {{Measurement of Stimulated Hawking
  Emission in an Analogue System}},\ }\href
  {https://doi.org/10.1103/PhysRevLett.106.021302} {\bibfield  {journal}
  {\bibinfo  {journal} {Phys. Rev. Lett.}\ }\textbf {\bibinfo {volume} {106}},\
  \bibinfo {pages} {021302} (\bibinfo {year} {2011})}\BibitemShut {NoStop}%
\bibitem [{\citenamefont {Euv\'e}\ \emph {et~al.}(2016)\citenamefont {Euv\'e},
  \citenamefont {Michel}, \citenamefont {Parentani}, \citenamefont {Philbin},\
  and\ \citenamefont {Rousseaux}}]{Euve}%
  \BibitemOpen
  \bibfield  {author} {\bibinfo {author} {\bibfnamefont {L.-P.}\ \bibnamefont
  {Euv\'e}}, \bibinfo {author} {\bibfnamefont {F.}~\bibnamefont {Michel}},
  \bibinfo {author} {\bibfnamefont {R.}~\bibnamefont {Parentani}}, \bibinfo
  {author} {\bibfnamefont {T.~G.}\ \bibnamefont {Philbin}},\ and\ \bibinfo
  {author} {\bibfnamefont {G.}~\bibnamefont {Rousseaux}},\ }\bibfield  {title}
  {\bibinfo {title} {{Observation of Noise Correlated by the Hawking Effect in
  a Water Tank}},\ }\href {https://doi.org/10.1103/PhysRevLett.117.121301}
  {\bibfield  {journal} {\bibinfo  {journal} {Phys. Rev. Lett.}\ }\textbf
  {\bibinfo {volume} {117}},\ \bibinfo {pages} {121301} (\bibinfo {year}
  {2016})}\BibitemShut {NoStop}%
\bibitem [{\citenamefont {Unruh}(2014)}]{Unruh2014}%
  \BibitemOpen
  \bibfield  {author} {\bibinfo {author} {\bibfnamefont {W.~G.}\ \bibnamefont
  {Unruh}},\ }\bibfield  {title} {\bibinfo {title} {{Has Hawking Radiation Been
  Measured?}},\ }\href {https://doi.org/10.1007/s10701-014-9778-0} {\bibfield
  {journal} {\bibinfo  {journal} {Foundations of Physics}\ }\textbf {\bibinfo
  {volume} {44}},\ \bibinfo {pages} {532} (\bibinfo {year} {2014})}\BibitemShut
  {NoStop}%
\bibitem [{\citenamefont {Garay}\ \emph {et~al.}(2000)\citenamefont {Garay},
  \citenamefont {Anglin}, \citenamefont {Cirac},\ and\ \citenamefont
  {Zoller}}]{PhysRevLett.85.4643}%
  \BibitemOpen
  \bibfield  {author} {\bibinfo {author} {\bibfnamefont {L.~J.}\ \bibnamefont
  {Garay}}, \bibinfo {author} {\bibfnamefont {J.~R.}\ \bibnamefont {Anglin}},
  \bibinfo {author} {\bibfnamefont {J.~I.}\ \bibnamefont {Cirac}},\ and\
  \bibinfo {author} {\bibfnamefont {P.}~\bibnamefont {Zoller}},\ }\bibfield
  {title} {\bibinfo {title} {{Sonic Analog of Gravitational Black Holes in
  Bose-Einstein Condensates}},\ }\href
  {https://doi.org/10.1103/PhysRevLett.85.4643} {\bibfield  {journal} {\bibinfo
   {journal} {Phys. Rev. Lett.}\ }\textbf {\bibinfo {volume} {85}},\ \bibinfo
  {pages} {4643} (\bibinfo {year} {2000})}\BibitemShut {NoStop}%
\bibitem [{\citenamefont {Barcel{\'{o}}}\ \emph {et~al.}(2001)\citenamefont
  {Barcel{\'{o}}}, \citenamefont {Liberati},\ and\ \citenamefont
  {Visser}}]{barcelo2001analogue}%
  \BibitemOpen
  \bibfield  {author} {\bibinfo {author} {\bibfnamefont {C.}~\bibnamefont
  {Barcel{\'{o}}}}, \bibinfo {author} {\bibfnamefont {S.}~\bibnamefont
  {Liberati}},\ and\ \bibinfo {author} {\bibfnamefont {M.}~\bibnamefont
  {Visser}},\ }\bibfield  {title} {\bibinfo {title} {{Analogue gravity from
  Bose-Einstein condensates}},\ }\href
  {https://doi.org/10.1088/0264-9381/18/6/312} {\bibfield  {journal} {\bibinfo
  {journal} {Classical and Quantum Gravity}\ }\textbf {\bibinfo {volume}
  {18}},\ \bibinfo {pages} {1137} (\bibinfo {year} {2001})}\BibitemShut
  {NoStop}%
\bibitem [{\citenamefont {Fischer}(2004)}]{ModPhysLett}%
  \BibitemOpen
  \bibfield  {author} {\bibinfo {author} {\bibfnamefont {U.~R.}\ \bibnamefont
  {Fischer}},\ }\bibfield  {title} {\bibinfo {title} {{Quasiparticle universes
  in Bose-Einstein condensates}},\ }\href
  {https://doi.org/10.1142/S0217732304015099} {\bibfield  {journal} {\bibinfo
  {journal} {Modern Physics Letters A}\ }\textbf {\bibinfo {volume} {19}},\
  \bibinfo {pages} {1789} (\bibinfo {year} {2004})}\BibitemShut {NoStop}%
\bibitem [{\citenamefont {Carusotto}\ \emph {et~al.}(2008)\citenamefont
  {Carusotto}, \citenamefont {Fagnocchi}, \citenamefont {Recati}, \citenamefont
  {Balbinot},\ and\ \citenamefont {Fabbri}}]{Carusotto_2008}%
  \BibitemOpen
  \bibfield  {author} {\bibinfo {author} {\bibfnamefont {I.}~\bibnamefont
  {Carusotto}}, \bibinfo {author} {\bibfnamefont {S.}~\bibnamefont
  {Fagnocchi}}, \bibinfo {author} {\bibfnamefont {A.}~\bibnamefont {Recati}},
  \bibinfo {author} {\bibfnamefont {R.}~\bibnamefont {Balbinot}},\ and\
  \bibinfo {author} {\bibfnamefont {A.}~\bibnamefont {Fabbri}},\ }\bibfield
  {title} {\bibinfo {title} {{Numerical observation of Hawking radiation from
  acoustic black holes in atomic Bose{\textendash}Einstein condensates}},\
  }\href {https://doi.org/10.1088/1367-2630/10/10/103001} {\bibfield  {journal}
  {\bibinfo  {journal} {New Journal of Physics}\ }\textbf {\bibinfo {volume}
  {10}},\ \bibinfo {pages} {103001} (\bibinfo {year} {2008})}\BibitemShut
  {NoStop}%
\bibitem [{\citenamefont {Steinhauer}(2016)}]{Steinhauer16}%
  \BibitemOpen
  \bibfield  {author} {\bibinfo {author} {\bibfnamefont {J.}~\bibnamefont
  {Steinhauer}},\ }\bibfield  {title} {\bibinfo {title} {{Observation of
  quantum Hawking radiation and its entanglement in an analogue black hole}},\
  }\href {http://dx.doi.org/10.1038/nphys3863} {\bibfield  {journal} {\bibinfo
  {journal} {Nat. Phys.}\ }\textbf {\bibinfo {volume} {12}},\ \bibinfo {pages}
  {959} (\bibinfo {year} {2016})}\BibitemShut {NoStop}%
\bibitem [{\citenamefont {Mu{\~n}oz~de Nova}\ \emph {et~al.}(2019)\citenamefont
  {Mu{\~n}oz~de Nova}, \citenamefont {Golubkov}, \citenamefont {Kolobov},\ and\
  \citenamefont {Steinhauer}}]{Munoz}%
  \BibitemOpen
  \bibfield  {author} {\bibinfo {author} {\bibfnamefont {J.~R.}\ \bibnamefont
  {Mu{\~n}oz~de Nova}}, \bibinfo {author} {\bibfnamefont {K.}~\bibnamefont
  {Golubkov}}, \bibinfo {author} {\bibfnamefont {V.~I.}\ \bibnamefont
  {Kolobov}},\ and\ \bibinfo {author} {\bibfnamefont {J.}~\bibnamefont
  {Steinhauer}},\ }\bibfield  {title} {\bibinfo {title} {{Observation of
  thermal Hawking radiation and its temperature in an analogue black hole}},\
  }\href {https://doi.org/10.1038/s41586-019-1241-0} {\bibfield  {journal}
  {\bibinfo  {journal} {Nature}\ }\textbf {\bibinfo {volume} {569}},\ \bibinfo
  {pages} {688} (\bibinfo {year} {2019})}\BibitemShut {NoStop}%
\bibitem [{\citenamefont {Kolobov}\ \emph {et~al.}(2021)\citenamefont
  {Kolobov}, \citenamefont {Golubkov}, \citenamefont {Mu{\~n}oz~de Nova},\ and\
  \citenamefont {Steinhauer}}]{Kolobov2021}%
  \BibitemOpen
  \bibfield  {author} {\bibinfo {author} {\bibfnamefont {V.~I.}\ \bibnamefont
  {Kolobov}}, \bibinfo {author} {\bibfnamefont {K.}~\bibnamefont {Golubkov}},
  \bibinfo {author} {\bibfnamefont {J.~R.}\ \bibnamefont {Mu{\~n}oz~de Nova}},\
  and\ \bibinfo {author} {\bibfnamefont {J.}~\bibnamefont {Steinhauer}},\
  }\bibfield  {title} {\bibinfo {title} {{Observation of stationary spontaneous
  Hawking radiation and the time evolution of an analogue black hole}},\ }\href
  {https://doi.org/10.1038/s41567-020-01076-0} {\bibfield  {journal} {\bibinfo
  {journal} {Nature Physics}\ }\textbf {\bibinfo {volume} {17}},\ \bibinfo
  {pages} {362} (\bibinfo {year} {2021})}\BibitemShut {NoStop}%
\bibitem [{\citenamefont {Philbin}\ \emph {et~al.}(2008)\citenamefont
  {Philbin}, \citenamefont {Kuklewicz}, \citenamefont {Robertson},
  \citenamefont {Hill}, \citenamefont {K{\"o}nig},\ and\ \citenamefont
  {Leonhardt}}]{Philbin1367}%
  \BibitemOpen
  \bibfield  {author} {\bibinfo {author} {\bibfnamefont {T.~G.}\ \bibnamefont
  {Philbin}}, \bibinfo {author} {\bibfnamefont {C.}~\bibnamefont {Kuklewicz}},
  \bibinfo {author} {\bibfnamefont {S.}~\bibnamefont {Robertson}}, \bibinfo
  {author} {\bibfnamefont {S.}~\bibnamefont {Hill}}, \bibinfo {author}
  {\bibfnamefont {F.}~\bibnamefont {K{\"o}nig}},\ and\ \bibinfo {author}
  {\bibfnamefont {U.}~\bibnamefont {Leonhardt}},\ }\bibfield  {title} {\bibinfo
  {title} {{Fiber-Optical Analog of the Event Horizon}},\ }\href
  {https://doi.org/10.1126/science.1153625} {\bibfield  {journal} {\bibinfo
  {journal} {Science}\ }\textbf {\bibinfo {volume} {319}},\ \bibinfo {pages}
  {1367} (\bibinfo {year} {2008})}\BibitemShut {NoStop}%
\bibitem [{\citenamefont {Drori}\ \emph {et~al.}(2019)\citenamefont {Drori},
  \citenamefont {Rosenberg}, \citenamefont {Bermudez}, \citenamefont
  {Silberberg},\ and\ \citenamefont {Leonhardt}}]{Rosenberg}%
  \BibitemOpen
  \bibfield  {author} {\bibinfo {author} {\bibfnamefont {J.}~\bibnamefont
  {Drori}}, \bibinfo {author} {\bibfnamefont {Y.}~\bibnamefont {Rosenberg}},
  \bibinfo {author} {\bibfnamefont {D.}~\bibnamefont {Bermudez}}, \bibinfo
  {author} {\bibfnamefont {Y.}~\bibnamefont {Silberberg}},\ and\ \bibinfo
  {author} {\bibfnamefont {U.}~\bibnamefont {Leonhardt}},\ }\bibfield  {title}
  {\bibinfo {title} {{Observation of Stimulated Hawking Radiation in an Optical
  Analogue}},\ }\href {https://doi.org/10.1103/PhysRevLett.122.010404}
  {\bibfield  {journal} {\bibinfo  {journal} {Phys. Rev. Lett.}\ }\textbf
  {\bibinfo {volume} {122}},\ \bibinfo {pages} {010404} (\bibinfo {year}
  {2019})}\BibitemShut {NoStop}%
\bibitem [{\citenamefont {Volovik}(1999)}]{Volovik_1999}%
  \BibitemOpen
  \bibfield  {author} {\bibinfo {author} {\bibfnamefont {G.~E.}\ \bibnamefont
  {Volovik}},\ }\bibfield  {title} {\bibinfo {title} {{Simulation of a
  Painlev{\'{e} }-Gullstrand black hole in a thin 3He-A film}},\ }\href
  {https://doi.org/10.1134/1.568079} {\bibfield  {journal} {\bibinfo  {journal}
  {Journal of Experimental and Theoretical Physics Letters}\ }\textbf {\bibinfo
  {volume} {69}},\ \bibinfo {pages} {705} (\bibinfo {year} {1999})}\BibitemShut
  {NoStop}%
\bibitem [{\citenamefont {Volovik}(2006)}]{Volovik2006}%
  \BibitemOpen
  \bibfield  {author} {\bibinfo {author} {\bibfnamefont {G.~E.}\ \bibnamefont
  {Volovik}},\ }\bibfield  {title} {\bibinfo {title} {{Horizons and Ergoregions
  in Superfluids}},\ }\href {https://doi.org/10.1007/s10909-006-9248-y}
  {\bibfield  {journal} {\bibinfo  {journal} {Journal of Low Temperature
  Physics}\ }\textbf {\bibinfo {volume} {145}},\ \bibinfo {pages} {337}
  (\bibinfo {year} {2006})}\BibitemShut {NoStop}%
\bibitem [{\citenamefont {Marino}(2008)}]{Marino}%
  \BibitemOpen
  \bibfield  {author} {\bibinfo {author} {\bibfnamefont {F.}~\bibnamefont
  {Marino}},\ }\bibfield  {title} {\bibinfo {title} {{Acoustic black holes in a
  two-dimensional ``photon fluid''}},\ }\href
  {https://doi.org/10.1103/PhysRevA.78.063804} {\bibfield  {journal} {\bibinfo
  {journal} {Phys. Rev. A}\ }\textbf {\bibinfo {volume} {78}},\ \bibinfo
  {pages} {063804} (\bibinfo {year} {2008})}\BibitemShut {NoStop}%
\bibitem [{\citenamefont {Nguyen}\ \emph {et~al.}(2015)\citenamefont {Nguyen},
  \citenamefont {Gerace}, \citenamefont {Carusotto}, \citenamefont {Sanvitto},
  \citenamefont {Galopin}, \citenamefont {Lema\^{\i}tre}, \citenamefont
  {Sagnes}, \citenamefont {Bloch},\ and\ \citenamefont {Amo}}]{Nguyen}%
  \BibitemOpen
  \bibfield  {author} {\bibinfo {author} {\bibfnamefont {H.~S.}\ \bibnamefont
  {Nguyen}}, \bibinfo {author} {\bibfnamefont {D.}~\bibnamefont {Gerace}},
  \bibinfo {author} {\bibfnamefont {I.}~\bibnamefont {Carusotto}}, \bibinfo
  {author} {\bibfnamefont {D.}~\bibnamefont {Sanvitto}}, \bibinfo {author}
  {\bibfnamefont {E.}~\bibnamefont {Galopin}}, \bibinfo {author} {\bibfnamefont
  {A.}~\bibnamefont {Lema\^{\i}tre}}, \bibinfo {author} {\bibfnamefont
  {I.}~\bibnamefont {Sagnes}}, \bibinfo {author} {\bibfnamefont
  {J.}~\bibnamefont {Bloch}},\ and\ \bibinfo {author} {\bibfnamefont
  {A.}~\bibnamefont {Amo}},\ }\bibfield  {title} {\bibinfo {title} {{Acoustic
  Black Hole in a Stationary Hydrodynamic Flow of Microcavity Polaritons}},\
  }\href {https://doi.org/10.1103/PhysRevLett.114.036402} {\bibfield  {journal}
  {\bibinfo  {journal} {Phys. Rev. Lett.}\ }\textbf {\bibinfo {volume} {114}},\
  \bibinfo {pages} {036402} (\bibinfo {year} {2015})}\BibitemShut {NoStop}%
\bibitem [{\citenamefont {Gerace}\ and\ \citenamefont
  {Carusotto}(2012)}]{PhysRevB.86.144505}%
  \BibitemOpen
  \bibfield  {author} {\bibinfo {author} {\bibfnamefont {D.}~\bibnamefont
  {Gerace}}\ and\ \bibinfo {author} {\bibfnamefont {I.}~\bibnamefont
  {Carusotto}},\ }\bibfield  {title} {\bibinfo {title} {{Analog Hawking
  radiation from an acoustic black hole in a flowing polariton superfluid}},\
  }\href {https://doi.org/10.1103/PhysRevB.86.144505} {\bibfield  {journal}
  {\bibinfo  {journal} {Phys. Rev. B}\ }\textbf {\bibinfo {volume} {86}},\
  \bibinfo {pages} {144505} (\bibinfo {year} {2012})}\BibitemShut {NoStop}%
\bibitem [{\citenamefont {Volovik}(2016)}]{VolovikWSM_2016}%
  \BibitemOpen
  \bibfield  {author} {\bibinfo {author} {\bibfnamefont {G.~E.}\ \bibnamefont
  {Volovik}},\ }\bibfield  {title} {\bibinfo {title} {{Black hole and Hawking
  radiation by type-II Weyl fermions}},\ }\href
  {https://doi.org/10.1134/S0021364016210050} {\bibfield  {journal} {\bibinfo
  {journal} {JETP Letters}\ }\textbf {\bibinfo {volume} {104}},\ \bibinfo
  {pages} {645} (\bibinfo {year} {2016})}\BibitemShut {NoStop}%
\bibitem [{\citenamefont {Kedem}\ \emph {et~al.}(2020)\citenamefont {Kedem},
  \citenamefont {Bergholtz},\ and\ \citenamefont {Wilczek}}]{WilczekWSM}%
  \BibitemOpen
  \bibfield  {author} {\bibinfo {author} {\bibfnamefont {Y.}~\bibnamefont
  {Kedem}}, \bibinfo {author} {\bibfnamefont {E.~J.}\ \bibnamefont
  {Bergholtz}},\ and\ \bibinfo {author} {\bibfnamefont {F.}~\bibnamefont
  {Wilczek}},\ }\bibfield  {title} {\bibinfo {title} {Black and white holes at
  material junctions},\ }\href
  {https://doi.org/10.1103/PhysRevResearch.2.043285} {\bibfield  {journal}
  {\bibinfo  {journal} {Phys. Rev. Res.}\ }\textbf {\bibinfo {volume} {2}},\
  \bibinfo {pages} {043285} (\bibinfo {year} {2020})}\BibitemShut {NoStop}%
\bibitem [{\citenamefont {Volovik}(2021)}]{VolovikWSM}%
  \BibitemOpen
  \bibfield  {author} {\bibinfo {author} {\bibfnamefont {G.~E.}\ \bibnamefont
  {Volovik}},\ }\bibfield  {title} {\bibinfo {title} {{Type-II Weyl Semimetal
  versus Gravastar}},\ }\href {https://doi.org/10.1134/S0021364021160013}
  {\bibfield  {journal} {\bibinfo  {journal} {JETP Letters}\ }\textbf {\bibinfo
  {volume} {114}},\ \bibinfo {pages} {236} (\bibinfo {year}
  {2021})}\BibitemShut {NoStop}%
\bibitem [{\citenamefont {Kerr}()}]{kerr2023blackholessingularities}%
  \BibitemOpen
  \bibfield  {author} {\bibinfo {author} {\bibfnamefont {R.~P.}\ \bibnamefont
  {Kerr}},\ }\href {https://arxiv.org/abs/2312.00841} {\bibinfo {title} {{Do
  Black Holes have Singularities?}}},\ \bibinfo {note} {preprint},\ \Eprint
  {https://arxiv.org/abs/2312.00841} {arXiv:2312.00841 [gr-qc]} \BibitemShut
  {NoStop}%
\bibitem [{\citenamefont {Penrose}(1969)}]{Penrose}%
  \BibitemOpen
  \bibfield  {author} {\bibinfo {author} {\bibfnamefont {R.}~\bibnamefont
  {Penrose}},\ }\bibfield  {title} {\bibinfo {title} {{{Gravitational Collapse:
  the Role of General Relativity}}},\ }\href
  {https://doi.org/10.1023/A:1016578408204} {\bibfield  {journal} {\bibinfo
  {journal} {Riv. Nuovo Cim.}\ }\textbf {\bibinfo {volume} {1}},\ \bibinfo
  {pages} {252} (\bibinfo {year} {1969})}\BibitemShut {NoStop}%
\bibitem [{\citenamefont {Volovik}(2023)}]{VolovikBriefReview}%
  \BibitemOpen
  \bibfield  {author} {\bibinfo {author} {\bibfnamefont {G.~E.}\ \bibnamefont
  {Volovik}},\ }\bibfield  {title} {\bibinfo {title} {{Gravity Through the
  Prism of Condensed Matter Physics (Brief Review)}},\ }\href
  {https://doi.org/10.1134/S0021364023602683} {\bibfield  {journal} {\bibinfo
  {journal} {JETP Letters}\ }\textbf {\bibinfo {volume} {118}},\ \bibinfo
  {pages} {531} (\bibinfo {year} {2023})}\BibitemShut {NoStop}%
\bibitem [{\citenamefont {Fischer}\ and\ \citenamefont
  {Datta}(2023)}]{PRD2023}%
  \BibitemOpen
  \bibfield  {author} {\bibinfo {author} {\bibfnamefont {U.~R.}\ \bibnamefont
  {Fischer}}\ and\ \bibinfo {author} {\bibfnamefont {S.}~\bibnamefont
  {Datta}},\ }\bibfield  {title} {\bibinfo {title} {Dispersive censor of
  acoustic spacetimes with a shock-wave singularity},\ }\href
  {https://doi.org/10.1103/PhysRevD.107.084023} {\bibfield  {journal} {\bibinfo
   {journal} {Phys. Rev. D}\ }\textbf {\bibinfo {volume} {107}},\ \bibinfo
  {pages} {084023} (\bibinfo {year} {2023})}\BibitemShut {NoStop}%
\bibitem [{\citenamefont {Barceló}\ \emph {et~al.}(2004)\citenamefont
  {Barceló}, \citenamefont {Liberati}, \citenamefont {Sonego},\ and\
  \citenamefont {Visser}}]{Barcelo_2004}%
  \BibitemOpen
  \bibfield  {author} {\bibinfo {author} {\bibfnamefont {C.}~\bibnamefont
  {Barceló}}, \bibinfo {author} {\bibfnamefont {S.}~\bibnamefont {Liberati}},
  \bibinfo {author} {\bibfnamefont {S.}~\bibnamefont {Sonego}},\ and\ \bibinfo
  {author} {\bibfnamefont {M.}~\bibnamefont {Visser}},\ }\bibfield  {title}
  {\bibinfo {title} {{Causal structure of analogue spacetimes}},\ }\href
  {https://doi.org/10.1088/1367-2630/6/1/186} {\bibfield  {journal} {\bibinfo
  {journal} {New Journal of Physics}\ }\textbf {\bibinfo {volume} {6}},\
  \bibinfo {pages} {186} (\bibinfo {year} {2004})}\BibitemShut {NoStop}%
\bibitem [{\citenamefont {Barcel\'o}\ \emph {et~al.}(2006)\citenamefont
  {Barcel\'o}, \citenamefont {Liberati}, \citenamefont {Sonego},\ and\
  \citenamefont {Visser}}]{Sonego}%
  \BibitemOpen
  \bibfield  {author} {\bibinfo {author} {\bibfnamefont {C.}~\bibnamefont
  {Barcel\'o}}, \bibinfo {author} {\bibfnamefont {S.}~\bibnamefont {Liberati}},
  \bibinfo {author} {\bibfnamefont {S.}~\bibnamefont {Sonego}},\ and\ \bibinfo
  {author} {\bibfnamefont {M.}~\bibnamefont {Visser}},\ }\bibfield  {title}
  {\bibinfo {title} {{Hawking-Like Radiation Does Not Require a Trapped
  Region}},\ }\href {https://doi.org/10.1103/PhysRevLett.97.171301} {\bibfield
  {journal} {\bibinfo  {journal} {Phys. Rev. Lett.}\ }\textbf {\bibinfo
  {volume} {97}},\ \bibinfo {pages} {171301} (\bibinfo {year}
  {2006})}\BibitemShut {NoStop}%
\bibitem [{\citenamefont {Datta}\ and\ \citenamefont
  {Fischer}(2022)}]{Datta_2022}%
  \BibitemOpen
  \bibfield  {author} {\bibinfo {author} {\bibfnamefont {S.}~\bibnamefont
  {Datta}}\ and\ \bibinfo {author} {\bibfnamefont {U.~R.}\ \bibnamefont
  {Fischer}},\ }\bibfield  {title} {\bibinfo {title} {{Analogue gravitational
  field from nonlinear fluid dynamics}},\ }\href
  {https://doi.org/10.1088/1361-6382/ac4828} {\bibfield  {journal} {\bibinfo
  {journal} {Classical and Quantum Gravity}\ }\textbf {\bibinfo {volume}
  {39}},\ \bibinfo {pages} {075018} (\bibinfo {year} {2022})}\BibitemShut
  {NoStop}%
\bibitem [{\citenamefont {Novello}\ \emph {et~al.}(2013)\citenamefont
  {Novello}, \citenamefont {Bittencourt}, \citenamefont {Moschella},
  \citenamefont {Goulart}, \citenamefont {Salim},\ and\ \citenamefont
  {Toniato}}]{Novello}%
  \BibitemOpen
  \bibfield  {author} {\bibinfo {author} {\bibfnamefont {M.}~\bibnamefont
  {Novello}}, \bibinfo {author} {\bibfnamefont {E.}~\bibnamefont
  {Bittencourt}}, \bibinfo {author} {\bibfnamefont {U.}~\bibnamefont
  {Moschella}}, \bibinfo {author} {\bibfnamefont {E.}~\bibnamefont {Goulart}},
  \bibinfo {author} {\bibfnamefont {J.~M.}\ \bibnamefont {Salim}},\ and\
  \bibinfo {author} {\bibfnamefont {J.~D.}\ \bibnamefont {Toniato}},\
  }\bibfield  {title} {\bibinfo {title} {Geometric scalar theory of gravity},\
  }\href {https://doi.org/10.1088/1475-7516/2013/06/014} {\bibfield  {journal}
  {\bibinfo  {journal} {Journal of Cosmology and Astroparticle Physics}\
  }\textbf {\bibinfo {volume} {2013}}\bibinfo  {number} { (06)},\ \bibinfo
  {pages} {014}}\BibitemShut {NoStop}%
\bibitem [{sup()}]{suppl}%
  \BibitemOpen
\bibfield  {number} {  }\href@noop {} {}\bibinfo {note} {The Supplemental
  Material, which quotes Refs.~[75-82], contains proofs of the theorems in the
  main text, the calculation of the event horizon radius, a discussion on
  hypersurfaces in the ABH spacetime, and a calculation of elapsed time and
  affine lengths during the travel of radial null rays from the trapped
  surface. Finally, it also contains a discussion of the null expansion and on
  singularity theorems in non-Einstein theories of gravity.}\BibitemShut
  {Stop}%
\bibitem [{\citenamefont {Chikkatur}\ \emph {et~al.}(2002)\citenamefont
  {Chikkatur}, \citenamefont {Shin}, \citenamefont {Leanhardt}, \citenamefont
  {Kielpinski}, \citenamefont {Tsikata}, \citenamefont {Gustavson},
  \citenamefont {Pritchard},\ and\ \citenamefont {Ketterle}}]{Chikka}%
  \BibitemOpen
  \bibfield  {author} {\bibinfo {author} {\bibfnamefont {A.~P.}\ \bibnamefont
  {Chikkatur}}, \bibinfo {author} {\bibfnamefont {Y.}~\bibnamefont {Shin}},
  \bibinfo {author} {\bibfnamefont {A.~E.}\ \bibnamefont {Leanhardt}}, \bibinfo
  {author} {\bibfnamefont {D.}~\bibnamefont {Kielpinski}}, \bibinfo {author}
  {\bibfnamefont {E.}~\bibnamefont {Tsikata}}, \bibinfo {author} {\bibfnamefont
  {T.~L.}\ \bibnamefont {Gustavson}}, \bibinfo {author} {\bibfnamefont {D.~E.}\
  \bibnamefont {Pritchard}},\ and\ \bibinfo {author} {\bibfnamefont
  {W.}~\bibnamefont {Ketterle}},\ }\bibfield  {title} {\bibinfo {title} {{A
  Continuous Source of Bose-Einstein Condensed Atoms}},\ }\href
  {https://doi.org/10.1126/science.296.5576.2193} {\bibfield  {journal}
  {\bibinfo  {journal} {Science}\ }\textbf {\bibinfo {volume} {296}},\ \bibinfo
  {pages} {2193} (\bibinfo {year} {2002})}\BibitemShut {NoStop}%
\bibitem [{\citenamefont {Holanda~Ribeiro}\ \emph {et~al.}(2022)\citenamefont
  {Holanda~Ribeiro}, \citenamefont {Baak},\ and\ \citenamefont
  {Fischer}}]{Caio}%
  \BibitemOpen
  \bibfield  {author} {\bibinfo {author} {\bibfnamefont {C.~C.}\ \bibnamefont
  {Holanda~Ribeiro}}, \bibinfo {author} {\bibfnamefont {S.-S.}\ \bibnamefont
  {Baak}},\ and\ \bibinfo {author} {\bibfnamefont {U.~R.}\ \bibnamefont
  {Fischer}},\ }\bibfield  {title} {\bibinfo {title} {{Existence of
  steady-state black hole analogs in finite quasi-one-dimensional Bose-Einstein
  condensates}},\ }\href {https://doi.org/10.1103/PhysRevD.105.124066}
  {\bibfield  {journal} {\bibinfo  {journal} {Phys. Rev. D}\ }\textbf {\bibinfo
  {volume} {105}},\ \bibinfo {pages} {124066} (\bibinfo {year}
  {2022})}\BibitemShut {NoStop}%
\bibitem [{\citenamefont {Grant}(2014)}]{grant2014theory}%
  \BibitemOpen
  \bibfield  {author} {\bibinfo {author} {\bibfnamefont {C.~P.}\ \bibnamefont
  {Grant}},\ }\href {https://books.google.fr/books?id=h3nZrQEACAAJ} {\emph
  {\bibinfo {title} {Theory of Ordinary Differential Equations}}}\ (\bibinfo
  {publisher} {CreateSpace Independent Publishing Platform},\ \bibinfo {year}
  {2014})\BibitemShut {NoStop}%
\bibitem [{\citenamefont {Teschl}()}]{teschlordinary}%
  \BibitemOpen
  \bibfield  {author} {\bibinfo {author} {\bibfnamefont {G.}~\bibnamefont
  {Teschl}},\ }\href {https://books.google.co.kr/books?id=FSObYfuWceMC} {\emph
  {\bibinfo {title} {{Ordinary Differential Equations and Dynamical
  Systems}}}},\ Graduate studies in mathematics\ (\bibinfo  {publisher}
  {American Mathematical Society, 2012})\BibitemShut {NoStop}%
\bibitem [{\citenamefont {Visser}\ and\ \citenamefont
  {Weinfurtner}(2005)}]{Visser_2005}%
  \BibitemOpen
  \bibfield  {author} {\bibinfo {author} {\bibfnamefont {M.}~\bibnamefont
  {Visser}}\ and\ \bibinfo {author} {\bibfnamefont {S.}~\bibnamefont
  {Weinfurtner}},\ }\bibfield  {title} {\bibinfo {title} {{Vortex analogue for
  the equatorial geometry of the Kerr black hole}},\ }\href
  {https://doi.org/10.1088/0264-9381/22/12/011} {\bibfield  {journal} {\bibinfo
   {journal} {Classical and Quantum Gravity}\ }\textbf {\bibinfo {volume}
  {22}},\ \bibinfo {pages} {2493} (\bibinfo {year} {2005})}\BibitemShut
  {NoStop}%
\bibitem [{\citenamefont {Strogatz}(2018)}]{strogatz2018nonlinear}%
  \BibitemOpen
  \bibfield  {author} {\bibinfo {author} {\bibfnamefont {S.~H.}\ \bibnamefont
  {Strogatz}},\ }\href@noop {} {\emph {\bibinfo {title} {{Nonlinear dynamics
  and chaos with student solutions manual: With applications to physics,
  biology, chemistry, and engineering}}}}\ (\bibinfo  {publisher} {CRC press},\
  \bibinfo {year} {2018})\BibitemShut {NoStop}%
\bibitem [{\citenamefont {Muchotrzeb-Czerny}(1986)}]{Muchotrzeb-Czerny1986}%
  \BibitemOpen
  \bibfield  {author} {\bibinfo {author} {\bibfnamefont {B.}~\bibnamefont
  {Muchotrzeb-Czerny}},\ }\bibfield  {title} {\bibinfo {title} {{Transonic
  accretion flow in a thin disk around a black hole. Part 3. Analytic
  considerations}},\ }\href
  {http://inis.iaea.org/search/search.aspx?orig_q=RN:21012349} {\bibfield
  {journal} {\bibinfo  {journal} {Acta Astronomica}\ }\textbf {\bibinfo
  {volume} {36}},\ \bibinfo {pages} {1} (\bibinfo {year} {1986})}\BibitemShut
  {NoStop}%
\bibitem [{\citenamefont {Das}(2002)}]{a1}%
  \BibitemOpen
  \bibfield  {author} {\bibinfo {author} {\bibfnamefont {T.~K.}\ \bibnamefont
  {Das}},\ }\bibfield  {title} {\bibinfo {title} {{Generalized Shock Solutions
  for Hydrodynamic Black Hole Accretion}},\ }\href
  {https://doi.org/10.1086/342114} {\bibfield  {journal} {\bibinfo  {journal}
  {The Astrophysical Journal}\ }\textbf {\bibinfo {volume} {577}},\ \bibinfo
  {pages} {880} (\bibinfo {year} {2002})}\BibitemShut {NoStop}%
\bibitem [{\citenamefont {Abraham}\ \emph {et~al.}(2006)\citenamefont
  {Abraham}, \citenamefont {Bilić},\ and\ \citenamefont {Das}}]{Abraham_2006}%
  \BibitemOpen
  \bibfield  {author} {\bibinfo {author} {\bibfnamefont {H.}~\bibnamefont
  {Abraham}}, \bibinfo {author} {\bibfnamefont {N.}~\bibnamefont {Bilić}},\
  and\ \bibinfo {author} {\bibfnamefont {T.~K.}\ \bibnamefont {Das}},\
  }\bibfield  {title} {\bibinfo {title} {{Acoustic horizons in axially
  symmetric relativistic accretion}},\ }\href
  {https://doi.org/10.1088/0264-9381/23/7/010} {\bibfield  {journal} {\bibinfo
  {journal} {Classical and Quantum Gravity}\ }\textbf {\bibinfo {volume}
  {23}},\ \bibinfo {pages} {2371} (\bibinfo {year} {2006})}\BibitemShut
  {NoStop}%
\bibitem [{\citenamefont {Bilić}\ \emph {et~al.}(2013)\citenamefont {Bilić},
  \citenamefont {Choudhary}, \citenamefont {Das},\ and\ \citenamefont
  {Nag}}]{Bilic_2014}%
  \BibitemOpen
  \bibfield  {author} {\bibinfo {author} {\bibfnamefont {N.}~\bibnamefont
  {Bilić}}, \bibinfo {author} {\bibfnamefont {A.}~\bibnamefont {Choudhary}},
  \bibinfo {author} {\bibfnamefont {T.~K.}\ \bibnamefont {Das}},\ and\ \bibinfo
  {author} {\bibfnamefont {S.}~\bibnamefont {Nag}},\ }\bibfield  {title}
  {\bibinfo {title} {{The role of axisymmetric flow configuration in the
  estimation of the analogue surface gravity and related Hawking-like
  temperature}},\ }\href {https://doi.org/10.1088/0264-9381/31/3/035002}
  {\bibfield  {journal} {\bibinfo  {journal} {Classical and Quantum Gravity}\
  }\textbf {\bibinfo {volume} {31}},\ \bibinfo {pages} {035002} (\bibinfo
  {year} {2013})}\BibitemShut {NoStop}%
\bibitem [{\citenamefont {Datta}(2016)}]{Datta2016}%
  \BibitemOpen
  \bibfield  {author} {\bibinfo {author} {\bibfnamefont {S.}~\bibnamefont
  {Datta}},\ }\bibfield  {title} {\bibinfo {title} {{Bondi flow revisited}},\
  }\href {https://doi.org/10.1007/s10509-016-2849-2} {\bibfield  {journal}
  {\bibinfo  {journal} {Astrophysics and Space Science}\ }\textbf {\bibinfo
  {volume} {361}},\ \bibinfo {pages} {260} (\bibinfo {year}
  {2016})}\BibitemShut {NoStop}%
\bibitem [{\citenamefont {Bloch}\ \emph {et~al.}(1999)\citenamefont {Bloch},
  \citenamefont {H\"ansch},\ and\ \citenamefont {Esslinger}}]{Blochcw}%
  \BibitemOpen
  \bibfield  {author} {\bibinfo {author} {\bibfnamefont {I.}~\bibnamefont
  {Bloch}}, \bibinfo {author} {\bibfnamefont {T.~W.}\ \bibnamefont
  {H\"ansch}},\ and\ \bibinfo {author} {\bibfnamefont {T.}~\bibnamefont
  {Esslinger}},\ }\bibfield  {title} {\bibinfo {title} {{Atom Laser with a cw
  Output Coupler}},\ }\href {https://doi.org/10.1103/PhysRevLett.82.3008}
  {\bibfield  {journal} {\bibinfo  {journal} {Phys. Rev. Lett.}\ }\textbf
  {\bibinfo {volume} {82}},\ \bibinfo {pages} {3008} (\bibinfo {year}
  {1999})}\BibitemShut {NoStop}%
\bibitem [{\citenamefont {Hagley}\ \emph {et~al.}(1999)\citenamefont {Hagley},
  \citenamefont {Deng}, \citenamefont {Kozuma}, \citenamefont {Wen},
  \citenamefont {Helmerson}, \citenamefont {Rolston},\ and\ \citenamefont
  {Phillips}}]{laserO}%
  \BibitemOpen
  \bibfield  {author} {\bibinfo {author} {\bibfnamefont {E.~W.}\ \bibnamefont
  {Hagley}}, \bibinfo {author} {\bibfnamefont {L.}~\bibnamefont {Deng}},
  \bibinfo {author} {\bibfnamefont {M.}~\bibnamefont {Kozuma}}, \bibinfo
  {author} {\bibfnamefont {J.}~\bibnamefont {Wen}}, \bibinfo {author}
  {\bibfnamefont {K.}~\bibnamefont {Helmerson}}, \bibinfo {author}
  {\bibfnamefont {S.~L.}\ \bibnamefont {Rolston}},\ and\ \bibinfo {author}
  {\bibfnamefont {W.~D.}\ \bibnamefont {Phillips}},\ }\bibfield  {title}
  {\bibinfo {title} {A well-collimated quasi-continuous atom laser},\ }\href
  {http://www.jstor.org/stable/2897503} {\bibfield  {journal} {\bibinfo
  {journal} {Science}\ }\textbf {\bibinfo {volume} {283}},\ \bibinfo {pages}
  {1706} (\bibinfo {year} {1999})}\BibitemShut {NoStop}%
\bibitem [{\citenamefont {Dalfovo}\ \emph {et~al.}(1999)\citenamefont
  {Dalfovo}, \citenamefont {Giorgini}, \citenamefont {Pitaevski\v\i},\ and\
  \citenamefont {Stringari}}]{RevModPhys.71.463}%
  \BibitemOpen
  \bibfield  {author} {\bibinfo {author} {\bibfnamefont {F.}~\bibnamefont
  {Dalfovo}}, \bibinfo {author} {\bibfnamefont {S.}~\bibnamefont {Giorgini}},
  \bibinfo {author} {\bibfnamefont {L.~P.}\ \bibnamefont {Pitaevski\v\i}},\
  and\ \bibinfo {author} {\bibfnamefont {S.}~\bibnamefont {Stringari}},\
  }\bibfield  {title} {\bibinfo {title} {{Theory of Bose-Einstein condensation
  in trapped gases}},\ }\href {https://doi.org/10.1103/RevModPhys.71.463}
  {\bibfield  {journal} {\bibinfo  {journal} {Rev. Mod. Phys.}\ }\textbf
  {\bibinfo {volume} {71}},\ \bibinfo {pages} {463} (\bibinfo {year}
  {1999})}\BibitemShut {NoStop}%
\bibitem [{\citenamefont {Pethick}\ and\ \citenamefont
  {Smith}(2002)}]{pethick2002bose}%
  \BibitemOpen
  \bibfield  {author} {\bibinfo {author} {\bibfnamefont {C.}~\bibnamefont
  {Pethick}}\ and\ \bibinfo {author} {\bibfnamefont {H.}~\bibnamefont
  {Smith}},\ }\href {https://books.google.co.kr/books?id=iBk0G3_5iIQC} {\emph
  {\bibinfo {title} {{Bose-Einstein Condensation in Dilute Gases}}}}\ (\bibinfo
   {publisher} {Cambridge University Press},\ \bibinfo {year}
  {2002})\BibitemShut {NoStop}%
\bibitem [{\citenamefont {Hawking}\ and\ \citenamefont
  {Ellis}(1973)}]{Hawking_Ellis_1973}%
  \BibitemOpen
  \bibfield  {author} {\bibinfo {author} {\bibfnamefont {S.~W.}\ \bibnamefont
  {Hawking}}\ and\ \bibinfo {author} {\bibfnamefont {G.~F.~R.}\ \bibnamefont
  {Ellis}},\ }\href@noop {} {\emph {\bibinfo {title} {{The Large Scale
  Structure of Space-Time}}}},\ Cambridge Monographs on Mathematical Physics\
  (\bibinfo  {publisher} {Cambridge University Press},\ \bibinfo {year}
  {1973})\BibitemShut {NoStop}%
\bibitem [{\citenamefont {Dadhich}()}]{dadhichderivation}%
  \BibitemOpen
  \bibfield  {author} {\bibinfo {author} {\bibfnamefont {N.}~\bibnamefont
  {Dadhich}},\ }\href {https://arxiv.org/abs/gr-qc/0511123} {\bibinfo {title}
  {{Derivation of the Raychaudhuri Equation}}},\ \bibinfo {note} {preprint},\
  \Eprint {https://arxiv.org/abs/gr-qc/0511123} {arXiv:gr-qc/0511123 [gr-qc]}
  \BibitemShut {NoStop}%
\bibitem [{\citenamefont {Kar}\ and\ \citenamefont {Sengupta}(2007)}]{Kar2007}%
  \BibitemOpen
  \bibfield  {author} {\bibinfo {author} {\bibfnamefont {S.}~\bibnamefont
  {Kar}}\ and\ \bibinfo {author} {\bibfnamefont {S.}~\bibnamefont {Sengupta}},\
  }\bibfield  {title} {\bibinfo {title} {{The Raychaudhuri equations: A brief
  review}},\ }\href {https://doi.org/10.1007/s12043-007-0110-9} {\bibfield
  {journal} {\bibinfo  {journal} {Pramana}\ }\textbf {\bibinfo {volume} {69}},\
  \bibinfo {pages} {49} (\bibinfo {year} {2007})}\BibitemShut {NoStop}%
\bibitem [{\citenamefont {Carroll}(2019)}]{carroll_2019}%
  \BibitemOpen
  \bibfield  {author} {\bibinfo {author} {\bibfnamefont {S.~M.}\ \bibnamefont
  {Carroll}},\ }\href {https://doi.org/10.1017/9781108770385} {\emph {\bibinfo
  {title} {Spacetime and Geometry: An Introduction to General Relativity}}}\
  (\bibinfo  {publisher} {Cambridge University Press},\ \bibinfo {year}
  {2019})\BibitemShut {NoStop}%
\bibitem [{\citenamefont {Andrews}\ \emph {et~al.}(1997)\citenamefont
  {Andrews}, \citenamefont {Kurn}, \citenamefont {Miesner}, \citenamefont
  {Durfee}, \citenamefont {Townsend}, \citenamefont {Inouye},\ and\
  \citenamefont {Ketterle}}]{Andrews}%
  \BibitemOpen
  \bibfield  {author} {\bibinfo {author} {\bibfnamefont {M.~R.}\ \bibnamefont
  {Andrews}}, \bibinfo {author} {\bibfnamefont {D.~M.}\ \bibnamefont {Kurn}},
  \bibinfo {author} {\bibfnamefont {H.-J.}\ \bibnamefont {Miesner}}, \bibinfo
  {author} {\bibfnamefont {D.~S.}\ \bibnamefont {Durfee}}, \bibinfo {author}
  {\bibfnamefont {C.~G.}\ \bibnamefont {Townsend}}, \bibinfo {author}
  {\bibfnamefont {S.}~\bibnamefont {Inouye}},\ and\ \bibinfo {author}
  {\bibfnamefont {W.}~\bibnamefont {Ketterle}},\ }\bibfield  {title} {\bibinfo
  {title} {{Propagation of Sound in a Bose-Einstein Condensate}},\ }\href
  {https://doi.org/10.1103/PhysRevLett.79.553} {\bibfield  {journal} {\bibinfo
  {journal} {Phys. Rev. Lett.}\ }\textbf {\bibinfo {volume} {79}},\ \bibinfo
  {pages} {553} (\bibinfo {year} {1997})}\BibitemShut {NoStop}%
\bibitem [{\citenamefont {Eddington}(1924)}]{Eddington}%
  \BibitemOpen
  \bibfield  {author} {\bibinfo {author} {\bibfnamefont {A.~S.}\ \bibnamefont
  {Eddington}},\ }\bibfield  {title} {\bibinfo {title} {{A Comparison of
  Whitehead's and Einstein's Formul\ae}},\ }\href
  {https://doi.org/10.1038/113192a0} {\bibfield  {journal} {\bibinfo  {journal}
  {Nature}\ }\textbf {\bibinfo {volume} {113}},\ \bibinfo {pages} {192}
  (\bibinfo {year} {1924})}\BibitemShut {NoStop}%
\bibitem [{\citenamefont {Finkelstein}(1958)}]{Finkelstein}%
  \BibitemOpen
  \bibfield  {author} {\bibinfo {author} {\bibfnamefont {D.}~\bibnamefont
  {Finkelstein}},\ }\bibfield  {title} {\bibinfo {title} {{Past-Future
  Asymmetry of the Gravitational Field of a Point Particle}},\ }\href
  {https://doi.org/10.1103/PhysRev.110.965} {\bibfield  {journal} {\bibinfo
  {journal} {Phys. Rev.}\ }\textbf {\bibinfo {volume} {110}},\ \bibinfo {pages}
  {965} (\bibinfo {year} {1958})}\BibitemShut {NoStop}%
\bibitem [{\citenamefont {Visser}({\natexlab{a}})}]{rainbow}%
  \BibitemOpen
  \bibfield  {author} {\bibinfo {author} {\bibfnamefont {M.}~\bibnamefont
  {Visser}},\ }\href {https://arxiv.org/abs/0712.0810} {\bibinfo {title}
  {{Emergent rainbow spacetimes: Two pedagogical examples}}} ({\natexlab{a}}),\
  \bibinfo {note} {preprint},\ \Eprint {https://arxiv.org/abs/0712.0810}
  {arXiv:0712.0810 [gr-qc]} \BibitemShut {NoStop}%
\bibitem [{\citenamefont {Weinfurtner}\ \emph {et~al.}(2009)\citenamefont
  {Weinfurtner}, \citenamefont {Jain}, \citenamefont {Visser},\ and\
  \citenamefont {Gardiner}}]{rainbowWV}%
  \BibitemOpen
  \bibfield  {author} {\bibinfo {author} {\bibfnamefont {S.}~\bibnamefont
  {Weinfurtner}}, \bibinfo {author} {\bibfnamefont {P.}~\bibnamefont {Jain}},
  \bibinfo {author} {\bibfnamefont {M.}~\bibnamefont {Visser}},\ and\ \bibinfo
  {author} {\bibfnamefont {C.~W.}\ \bibnamefont {Gardiner}},\ }\bibfield
  {title} {\bibinfo {title} {{Cosmological particle production in emergent
  rainbow spacetimes}},\ }\href
  {http://stacks.iop.org/0264-9381/26/i=6/a=065012} {\bibfield  {journal}
  {\bibinfo  {journal} {Classical and Quantum Gravity}\ }\textbf {\bibinfo
  {volume} {26}},\ \bibinfo {pages} {065012} (\bibinfo {year}
  {2009})}\BibitemShut {NoStop}%
\bibitem [{\citenamefont {Jacobson}(1995)}]{Ted}%
  \BibitemOpen
  \bibfield  {author} {\bibinfo {author} {\bibfnamefont {T.}~\bibnamefont
  {Jacobson}},\ }\bibfield  {title} {\bibinfo {title} {{Thermodynamics of
  Spacetime: The Einstein Equation of State}},\ }\href
  {https://doi.org/10.1103/PhysRevLett.75.1260} {\bibfield  {journal} {\bibinfo
   {journal} {Phys. Rev. Lett.}\ }\textbf {\bibinfo {volume} {75}},\ \bibinfo
  {pages} {1260} (\bibinfo {year} {1995})}\BibitemShut {NoStop}%
\bibitem [{\citenamefont {Coddington}\ and\ \citenamefont
  {Levinson}(1984)}]{coddington1984theory}%
  \BibitemOpen
  \bibfield  {author} {\bibinfo {author} {\bibfnamefont {A.}~\bibnamefont
  {Coddington}}\ and\ \bibinfo {author} {\bibfnamefont {N.}~\bibnamefont
  {Levinson}},\ }\href {https://books.google.co.kr/books?id=AUAbvgAACAAJ}
  {\emph {\bibinfo {title} {{Theory of Ordinary Differential Equations}}}},\
  International series in pure and applied mathematics\ (\bibinfo  {publisher}
  {R.E. Krieger},\ \bibinfo {year} {1984})\BibitemShut {NoStop}%
\bibitem [{\citenamefont {Agarwal}\ \emph {et~al.}(2012)\citenamefont
  {Agarwal}, \citenamefont {Das}, \citenamefont {Dey},\ and\ \citenamefont
  {Nag}}]{Agarwal2012}%
  \BibitemOpen
  \bibfield  {author} {\bibinfo {author} {\bibfnamefont {S.}~\bibnamefont
  {Agarwal}}, \bibinfo {author} {\bibfnamefont {T.~K.}\ \bibnamefont {Das}},
  \bibinfo {author} {\bibfnamefont {R.}~\bibnamefont {Dey}},\ and\ \bibinfo
  {author} {\bibfnamefont {S.}~\bibnamefont {Nag}},\ }\bibfield  {title}
  {\bibinfo {title} {{An analytical study on the multi-critical behaviour and
  related bifurcation phenomena for relativistic black hole accretion}},\
  }\href {https://doi.org/10.1007/s10714-012-1358-z} {\bibfield  {journal}
  {\bibinfo  {journal} {General Relativity and Gravitation}\ }\textbf {\bibinfo
  {volume} {44}},\ \bibinfo {pages} {1637} (\bibinfo {year}
  {2012})}\BibitemShut {NoStop}%
\bibitem [{\citenamefont {Hawking}(1965)}]{Hawking0}%
  \BibitemOpen
  \bibfield  {author} {\bibinfo {author} {\bibfnamefont {S.~W.}\ \bibnamefont
  {Hawking}},\ }\bibfield  {title} {\bibinfo {title} {Occurrence of
  singularities in open universes},\ }\href
  {https://doi.org/10.1103/PhysRevLett.15.689} {\bibfield  {journal} {\bibinfo
  {journal} {Phys. Rev. Lett.}\ }\textbf {\bibinfo {volume} {15}},\ \bibinfo
  {pages} {689} (\bibinfo {year} {1965})}\BibitemShut {NoStop}%
\bibitem [{\citenamefont {Hawking}\ and\ \citenamefont
  {Bondi}(1966{\natexlab{a}})}]{Hawking1}%
  \BibitemOpen
  \bibfield  {author} {\bibinfo {author} {\bibfnamefont {S.~W.}\ \bibnamefont
  {Hawking}}\ and\ \bibinfo {author} {\bibfnamefont {H.}~\bibnamefont
  {Bondi}},\ }\bibfield  {title} {\bibinfo {title} {{The occurrence of
  singularities in cosmology}},\ }\href
  {https://doi.org/10.1098/rspa.1966.0221} {\bibfield  {journal} {\bibinfo
  {journal} {Proceedings of the Royal Society of London. Series A. Mathematical
  and Physical Sciences}\ }\textbf {\bibinfo {volume} {294}},\ \bibinfo {pages}
  {511} (\bibinfo {year} {1966}{\natexlab{a}})}\BibitemShut {NoStop}%
\bibitem [{\citenamefont {Hawking}\ and\ \citenamefont
  {Bondi}(1966{\natexlab{b}})}]{Hawking2}%
  \BibitemOpen
  \bibfield  {author} {\bibinfo {author} {\bibfnamefont {S.~W.}\ \bibnamefont
  {Hawking}}\ and\ \bibinfo {author} {\bibfnamefont {H.}~\bibnamefont
  {Bondi}},\ }\bibfield  {title} {\bibinfo {title} {{The occurrence of
  singularities in cosmology. II}},\ }\href
  {https://doi.org/10.1098/rspa.1966.0255} {\bibfield  {journal} {\bibinfo
  {journal} {Proceedings of the Royal Society of London. Series A. Mathematical
  and Physical Sciences}\ }\textbf {\bibinfo {volume} {295}},\ \bibinfo {pages}
  {490} (\bibinfo {year} {1966}{\natexlab{b}})}\BibitemShut {NoStop}%
\bibitem [{\citenamefont {Hawking}(1967)}]{Hawking3}%
  \BibitemOpen
  \bibfield  {author} {\bibinfo {author} {\bibfnamefont {S.~W.}\ \bibnamefont
  {Hawking}},\ }\bibfield  {title} {\bibinfo {title} {{The Occurrence of
  Singularities in Cosmology. III. Causality and Singularities}},\ }\href
  {http://www.jstor.org/stable/2415769} {\bibfield  {journal} {\bibinfo
  {journal} {Proceedings of the Royal Society of London. Series A, Mathematical
  and Physical Sciences}\ }\textbf {\bibinfo {volume} {300}},\ \bibinfo {pages}
  {187} (\bibinfo {year} {1967})}\BibitemShut {NoStop}%
\bibitem [{\citenamefont {Hawking}\ and\ \citenamefont
  {Penrose}(1970)}]{Hawking1970singularities}%
  \BibitemOpen
  \bibfield  {author} {\bibinfo {author} {\bibfnamefont {S.~W.}\ \bibnamefont
  {Hawking}}\ and\ \bibinfo {author} {\bibfnamefont {R.}~\bibnamefont
  {Penrose}},\ }\bibfield  {title} {\bibinfo {title} {{The singularities of
  gravitational collapse and cosmology}},\ }\href
  {https://doi.org/10.1098/rspa.1970.0021} {\bibfield  {journal} {\bibinfo
  {journal} {Proceedings of the Royal Society of London. A. Mathematical and
  Physical Sciences}\ }\textbf {\bibinfo {volume} {314}},\ \bibinfo {pages}
  {529} (\bibinfo {year} {1970})}\BibitemShut {NoStop}%
\bibitem [{\citenamefont {Visser}({\natexlab{b}})}]{visseraffine}%
  \BibitemOpen
  \bibfield  {author} {\bibinfo {author} {\bibfnamefont {M.}~\bibnamefont
  {Visser}},\ }\href {https://arxiv.org/abs/2211.07835} {\bibinfo {title}
  {{Efficient computation of null affine parameters}}} ({\natexlab{b}}),\
  \bibinfo {note} {preprint},\ \Eprint {https://arxiv.org/abs/2211.07835}
  {arXiv:2211.07835 [gr-qc]} \BibitemShut {NoStop}%
\end{thebibliography}%



\vspace*{80em} 
\newpage
\begin{widetext}
\setcounter{equation}{0}
\setcounter{figure}{0}
\setcounter{table}{0}
\setcounter{page}{1}
\renewcommand{\theequation}{S\arabic{equation}}
\renewcommand{\thefigure}{S\arabic{figure}}

\section{Supplemental Material}
\subsection{Proof of main text theorems}

{\bf{Theorem 1.}} Given an IVP at $r=r_0$ with finite $\rho_{0}(r_0)>0$, and finite $ v^r_{0}(r_0)$ as well as finite $v^\phi_{0}(r_0)$, the maximal domain in $r$ of an axisymmetric steady physical flow ($M_P$) with nonzero mass flux rate is the region where $f(c_{s0}, v^r_{0}, r)$ and $g(c_{s0}, v^r_{0},r)$ exist and are Lipschitz continuous in $v^r_{0}$, and in $c_{s0}$ when $V_{\rm ext}(r)$ is smooth in $M$.

\begin{proof}
We are given an IVP as the theorem states; $C_1, C_2, l$ are finite constants. Constant finite $C_2$, and constant finite mass flow rate $C_1$ asserts that $\rho_{0}, v^r_{0}$ does not change sign. Therefore, if the IVP of the ODEs (Eq. \eqref{dcsdr}, Eq. \eqref{dvdr}) specifies positive $\rho_{0}(r_0)$ with finite velocity components, the flow remains physical 
 for the same radial direction (radially outward or inward flow) everywhere in the $r$-domain of a unique solution. Thus we ensured that the flow remains physical whenever there exists a unique solution from the IVP. The $r$-domain of unique solution is determined from Lipschitz continuity as the theorem states.  
If the IVP at a given radius $r_0$ is such that $\frac{dv_{0}}{dr}$ is infinity, Lipschitz continuity is not satisfied. We see from Eq.~\eqref{dvdr} that this is a possibility. If at $r=r_0$, $c_{s0}(r_0)= v^r_{0}(r_0)$ and if the numerator of the expression is nonzero then  $f(c_{s0}, v^r_{0}, r)$, and $g(c_{s0}, v^r_{0},r)$ are diverging, and no Lipschitz continuity in $c_{s0}$ and in $v^r_{0}$ is obtained. From Picard-Lindel\"{o}f's uniqueness theorem  \cite{coddington1984theory}  
for such an IVP, the solution is not unique in the annular $\epsilon$-neighborhood $[r_0-\epsilon, r_0+\epsilon]$ ($\epsilon >0$) given a smooth $V_{\rm ext}(r)$ in $M$. If at some $r$, $\frac{dv_{0}}{dr}$ is of the $0/0$ form, the functions $f(c_{s0}, v^r_{0}, r)$, and $g(c_{s0}, v^r_{0},r)$ do not exist, and we need to work with limiting values by finding $\frac{dv_{0}}{dr}$ using L'H\^{o}pital's rule.
\end{proof}
{\bf{Theorem 2.}} {If an axisymmetric steady physical flow  with nonzero mass flow rate for a smooth $V_{\rm ext}(r)$ exists in $M$, then $\lim_{r\to 0}V_{\rm ext}(r) =-\infty$.}
\begin{proof}
Let us assume $\lim_{r\to 0}V_{\rm ext}(r)$ is finite. For a physical flow with nonzero mass flow rate ($C_1$ is then finite and positive), as well as the integration constants $C_2$ and $l$ are finite.

$l \neq 0$ is a trivial case. For $l\neq 0$, the limiting value of the angular momentum term of the Bernoulli constant is $\infty$ as $r$ approaches 0.  Since $C_2$ is finite, therefore $\lim_{r\rightarrow 0}V_{\rm ext}(r)$ is finite, which contradicts our initial assumption. 
 
For $l=0$, the expression for $C_1$  along with the expression for the sound speed, gives us
\begin{equation}\label{csv}
c_{s0}^2  (v_{0}^r)^{\gamma -1}  = \frac{C_1'}{r^{\gamma -1}},
\end{equation}
where $C_1'$ is a finite and positive constant $\because K>0$.
Since $\gamma >1$, Eq. \eqref{csv} confirms that the limiting values of 
$c_{s0} (r)$ or $v_{0}^r$ are $\infty$ as $r$ approaches zero.
On the other hand, since $\lim_{r\to 0}V_{\rm ext}(r)$ is assumed to be finite, finite $C_2$ for $l=0$ confirms that both $c_{s0}(r)$ and $v_{0}(r)$ are finite as $r$ approaches zero, contradicting Eq. \eqref{csv}.
Therefore to balance the quantities in the Bernoulli constant $C_2$,
$\lim_{r\to 0}V_{\rm ext}(r)=-\infty$.
\end{proof}
\subsection{Hypersurfaces in a $2+1$D ABH geometry}
Any $2$D smooth hypersurface in a $2+1$D manifold in cylindrical polar coordinates is in general, in a coordinate basis for which $r$ is a position coordinate, defined by 
\begin{equation}
f(t,r,\phi)={\rm constant}.
\end{equation}
Evidently, the tangent vectors $\frac{dx^\mu}{d\lambda}$ ($\lambda$ is parametrizing the tangent curve) at any point on the hypersurface satisfies the equation 
\begin{equation}\label{tangent}
\partial _\mu f \frac{dx^\mu}{d\lambda} =0.
\end{equation}
The above Eq.~\eqref{tangent} 
implies that at any point on the hypersurface $\partial _\mu f$ is normal to the hypersurface.
One interesting 2D hypersurface characterizing the spacetime with the metric \eqref{agmn} is the event horizon. Generalizing the definition of event horizon \cite{carroll_2019} in $3+1$D to $2+1$D, we define it as the $f(t,r,\phi)=r=$constant hypersurface where the gradient $\partial_\mu f =\partial _\mu r$ is null, i.e., $\partial _\mu r \partial ^\mu r =0$. Comparing with Eq.~\eqref{tangent}, we see that at the event horizon, $\partial ^\mu r$ lies on the $r={\rm constant}$ hypersurface, i.e., $\partial ^\mu r \propto \frac{dx^\mu}{d\lambda}$. Therefore, to move on the event horizon (1D circle in $2+1$D geometry) hypersurface, the corresponding trajectory in spacetime must be soundlike. $\partial _\mu r =(0,1,0)$, therefore $\partial _\mu r \partial^\mu r =0$ $\Rightarrow g^{rr}=0$. Finally, the condition for an event horizon at $r=r_H$ of the metric \eqref{agmn}  ($\rho_{0}^{-1}(r_H)\neq 0$ for physical flows), thus reads 
\begin{equation}\label{H1}
(c_{s0}(r_H))^2=(v_{0}^{r}(r_H))^2.
\end{equation}
Since we have a time-independent metric,  the time translation vector $\partial_t$ is a Killing vector $K$, i.e., $K^\mu =(1,0,0)$. The stationary limit surface (or ergosurface) is defined as the surface where $K$ is null \cite{carroll_2019}, therefore $K_\mu K^\mu =0$. Hence the stationary limit surface is determined by the condition $g_{tt} =0$. From  Eq.~\eqref{agmn}, we obtain the condition for a stationary limit surface at $r=r_s$ ($\left(\frac{\rho_{0}}{c_{s0}}\right)^2\neq 0$ within the maximal domain of physical flows): 
\begin{equation}\label{S}
(c_{s0}(r_s))^{2}=(v_{0}^r(r_s))^2+\frac{l^2}{r_s^2}.
\end{equation}
Similar to the definition of a stationary limit surface in $3+1$D, in $2+1$D it is a  hypersurface where to stay at a fixed coordinate position $(r,\phi)$, the corresponding trajectory in spacetime must be soundlike (in usual relativity terminology this of corresponds to soundlike). The difference between a static limit surface and a event horizon can be explained as follows. In the case of a black hole event horizon, $r=$ constant sphere/circle (in 3+1D or 2+1D) is soundlike, i.e., even if the observer rotates on the sphere/circle, still they can not be timelike, or still they can not escape the event horizon; the centrifugal force of rotation can not win against gravity at the horizon. However, in the case of static limit surface, the time axis (movement along time axis) is soundlike, i.e., to stay at a fixed coordinate position, the observer can no more be timelike. To stay at a fixed position in gravity without rotation, we need a force opposing gravity, say the electromagnetic repulsion of the horizontal ground to a person standing on the ground, or if a person wants to float at a fixed point without moving in air, they would need some propeller system to provide the upward thrust. Naturally, with the aid of the centrifugal force, it is easier to stay at a particular radius than to stay at a fixed point in space when there is no ground to stand on. In Newtonian gravity, if the propeller of the observer is strong enough, it is possible to float at a fixed point in space. In the ergoregion in a rotational black hole, this is not possible because of frame dragging. Nevertheless, an object can obtain negative energy (with respect to an observer at infinity) inside the ergoregion, and escape the black hole Ergo region (Penrose process). Therefore the event horizon imposes a stronger condition than a static limit surface in a rotating black hole. An observer or light ray, once inside the ergoregion, can escape from the static limit surface, but once inside the event horizon, no escape is possible anymore.
\medskip

{\bf Theorem 3.} {For an analogue black or white hole spacetime with both an event horizon and a static limit surface, $r_s\geq r_H$.}

\begin{proof}
For an axisymmetric flow to represent an analogue black hole or white hole spacetime, the following conditions have to hold: In the large $r$ region within the domain of physical flow (a) $(c_{s0}(r))^2\ge (v^r_{0}(r))^2$, (b) $(c_{s0}(r))^2 \ge (v^\phi _{0}(r))^2$, where $v_{0}^\phi$ is $\frac{l}{r}$; and in the small $r$ region within the domain of physical flow (c) $(c_{s0}(r))^2\le (v_{0}^{r}(r))^2$, (d) $(c_{s0}(r))^{2}\le (v_{0}^r(r))^2+\frac{l^2}{r^2}$.
The conditions (a) and (c) affirm the existence of an event horizon of finite radius because all the flow variables are continuous functions of $r$, and inside the event horizon the normal vector to a $r=$ constant hypersurface, $\partial _\mu r$, is timelike. The conditions (b) and (d) affirm the existence of an ergosurface of finite radius. Inside the ergosphere, $K^\mu$ is spacelike. With conditions (a)-(d) being satisfied, radially inward or outward flow represents the analogue black hole/white hole space time.
In the small $r$ region where the condition (c) is satisfied, the static limit surface condition can not be satisfied. 
Therefore, $r_s \geq r_H$. Equality holds for $l=0$.
\end{proof} 
The region between $r=r_s$ and $r=r_H$ is the {\it ergosphere}. In a constant time-slice embedding, it is a two-dimensional annular region in a $2+1$D ABH.

\subsection{Calculating the location of the analogue black hole event horizon and Mach number profile}
The function $g(c_{s0}, v^r_{0}, r)$ is not Lipschitz continuous at the event horizon, but there is a way to tackle the problem. In this section, we calculate the event horizon location in our axisymmetric steady flow model with nonzero mass flow rate by following an analytical technique familiar in the astrophysical literature \cite{Muchotrzeb-Czerny1986, a1, Abraham_2006,Bilic_2014, Datta2016}.
For a finite $\frac{dv_{0}^r}{dr}$ at the horizon ($r=r_H$) of the analogue black hole where $v_{0}^r=c_{s0}=c_{sH}$, in Eq.~\eqref{dvdr} we set the numerator to zero so that L'H\^{o}pital's rule can be applied. We then have
\begin{equation}\label{tbr}
V_{\rm ext}^{(1)}(r_H) r_H^3- c_{sH}^2 r_H^2-l^2=0.
\end{equation}
where $V_{\rm ext}^{(1)}$ is the first derivative of $V_{\rm ext}(r)$ with respect to $r$.
Bernoulli's constant at $r=r_H$ gives us,
\begin{equation}\label{csr}
c_{sH}^2=\frac{2(\gamma -1) C_2}{(\gamma +1)}-\frac{(\gamma -1)}{(\gamma +1)}\frac{l^2}{r_H^2}-\frac{2(\gamma -1) V_{\rm ext}(r_H)}{(\gamma +1)}.
\end{equation}
Using Eq. \eqref{csr} and Eq. \eqref{tbr}, we find an equation for $r_H$ with 
a given $V_{\rm ext} (r)$, 
\begin{equation}\label{rhpol}
V_{\rm ext}^{(1)}(r_H) r_H^3-\frac{2(\gamma -1)}{(\gamma +1)}\left(C_2 -V_{\rm ext}(r_H)\right)r_H^2
-\frac{2}{(\gamma +1)}l^2=0.
\end{equation}
Therefore, for a given $V_{\rm ext}$, $r_H=r_H (l, C_2, \gamma)$. Naturally, for an arbitrary $V_{\rm ext}(r)$, multiple solutions of $r_H$ can exist. The reference \cite{Agarwal2012} provides a mathematical study (using the  theory of polynomials) of the dependence of a multi-transonic flow on the parameter space $(l, C_2, \gamma)$ for astrophysical choices of $V_{\rm ext}(r)$.
With this solution of $r_H$ from Eq. \eqref{rhpol}, for a physical flow, $c_{sH}$ has to be real and positive. From Eq. \eqref{tbr}, we derive the following condition:
\begin{equation}\label{ineq}
c_{sH}^2=r_HV_{\rm ext}^{(1)}(r_H)-\frac{l^2}{r_H^2}>0.
\end{equation}
We apply L'H\^{o}pital's rule in Eq.~\eqref{dvdr} to find $\frac{dv_{0}^r}{dr}|_{r=r_H}=q$. We are left with a simple quadratic equation for $q$, 
\begin{equation}\label{qpol}
q^2 + Bq + C=0,
\end{equation} 
 where 
\begin{eqnarray}\label{BC}
& B= \frac{2(\gamma -1)c_{sH}}{r_H},\\
& C=\frac{\frac{\gamma c_{sH}^2}{r_H ^2}+\frac{3l^2}{r_H^4}+V_{\rm ext}^{(2)}(r_H)}{(\gamma +1)},
\end{eqnarray}
$V_{\rm ext}^{(2)}$ denotes the second derivative of $V_{\rm ext}(r)$ with respect to $r$. 
$\Delta=B^2 -4C$ has to be greater than zero so that $q$ is real. Two roots of the quadratic equation \eqref{qpol} lead to two solutions of flow variables with the same initial condition at $r=r_H$.  
For $r_H> 0$, and for real $c_{sH}$ (i.e., inequality Eq.~\eqref{ineq} holds) from the expression of $B$, we see $B>0$. 
$\Delta =0$ is a borderline case between spiral ($\Delta <0$) and nodal or saddle type fixed points ( for nodal or saddle type, $\Delta >0$). Eq. \eqref{continuity}, and Eq. \eqref{vr} near $r=r_H$ can be mapped to a dynamical system with $r=r_H$ being the fixed point \cite{Muchotrzeb-Czerny1986}. For $\Delta > 0$, $C=0$ is the borderline case between saddle type ($C<0$) and nodal type ($C>0$). For a dynamical system, a  linear stability analysis of fixed point fails for the borderline cases \cite{strogatz2018nonlinear}, and one needs to study higher order behavior. $\Delta$ and $C$ in terms of $r_H$ for a given $V_{\rm ext} (r)$ are given by
\begin{eqnarray}
& \Delta = \frac{4}{(\gamma +1)}\left[-V_{\rm ext}^{(2)}(r_H)+\frac{V_{\rm ext}^{(1)}(r_H)}{r_H}\frac{(1-3\gamma)}{(\gamma +1)}-\frac{4l ^2}{r_H^4(\gamma +1)}\right]\nonumber\\
& \\
& C=\frac{1}{(\gamma +1)}\left[ \frac{\gamma}{r_H^2}\left(-\frac{l^2}{r_H^2} +r_H V_{\rm ext}^{(1)}(r_H)\right)+\frac{3l^2}{r_H^4}+V_{\rm ext}^{(2)}(r_H)\right]\nonumber\\
\end{eqnarray}
\begin{corollary}\label{cor2}
If at the lowest positive root $r_H$ of the Eq. \eqref{rhpol}, if the condition \eqref{ineq} is satisfied and if $\Delta >0$ then $\lim _{r\to 0}V_{\rm ext}(r)=-\infty$. 
\end{corollary}
\begin{proof}
If at the lowest positive root $r_H$ of the Eq. \eqref{rhpol}, if the condition \eqref{ineq} is satisfied and if $\Delta >0$, an analogue black hole solution exists arbitrarily close to $r=0$; and theorem \ref{thm2} confirms that such a situation occurs only if $\lim _{r\to 0}V_{\rm ext}(r)=-\infty$. 
\end{proof}
 Therefore, choosing an appropriate potential is essential to create an analogue black hole from an axisymmetric flow. 
The simplest possibility of such an appropriate external potential could be 
\begin{equation}
V_{\rm ext}(r)=-\frac{V_0}{r^\alpha},
\end{equation}
where $\alpha >0$, $V_0>0$. We note that this external potential satisfies theorem \ref{thm2}.
We consider $l=0$ first. Therefore, from Eq. \eqref{rhpol},
\begin{eqnarray}\label{rhal}
r_H^\alpha &=\frac{\left((\gamma +1)\alpha-2(\gamma -1)\right)V_0}{2(\gamma -1)C_2}. \\
& \Delta > 0\Rightarrow \alpha >\frac{2(\gamma -1)}{(\gamma +1)}.\label{alga}\\
& C<0 \Rightarrow \alpha > \gamma -1.\label{CC}
\end{eqnarray}
For a BEC ($\gamma =2$), $\Delta > 0 $, $C=0$, we have to rely on numerics to analyze the near horizon behavior, by finding two roots at of $\frac{d v_{(0)}^r}{dr}$ at the horizon, 
cf.~Eq.~\eqref{qpol}. With these initial values, we numerically find the Mach number profile, cf.~Fig. \ref{Mar}.
\begin{figure}[t]
\vspace*{0.5em}
\centering
\hspace*{-0.5em}
\includegraphics[scale=0.2]{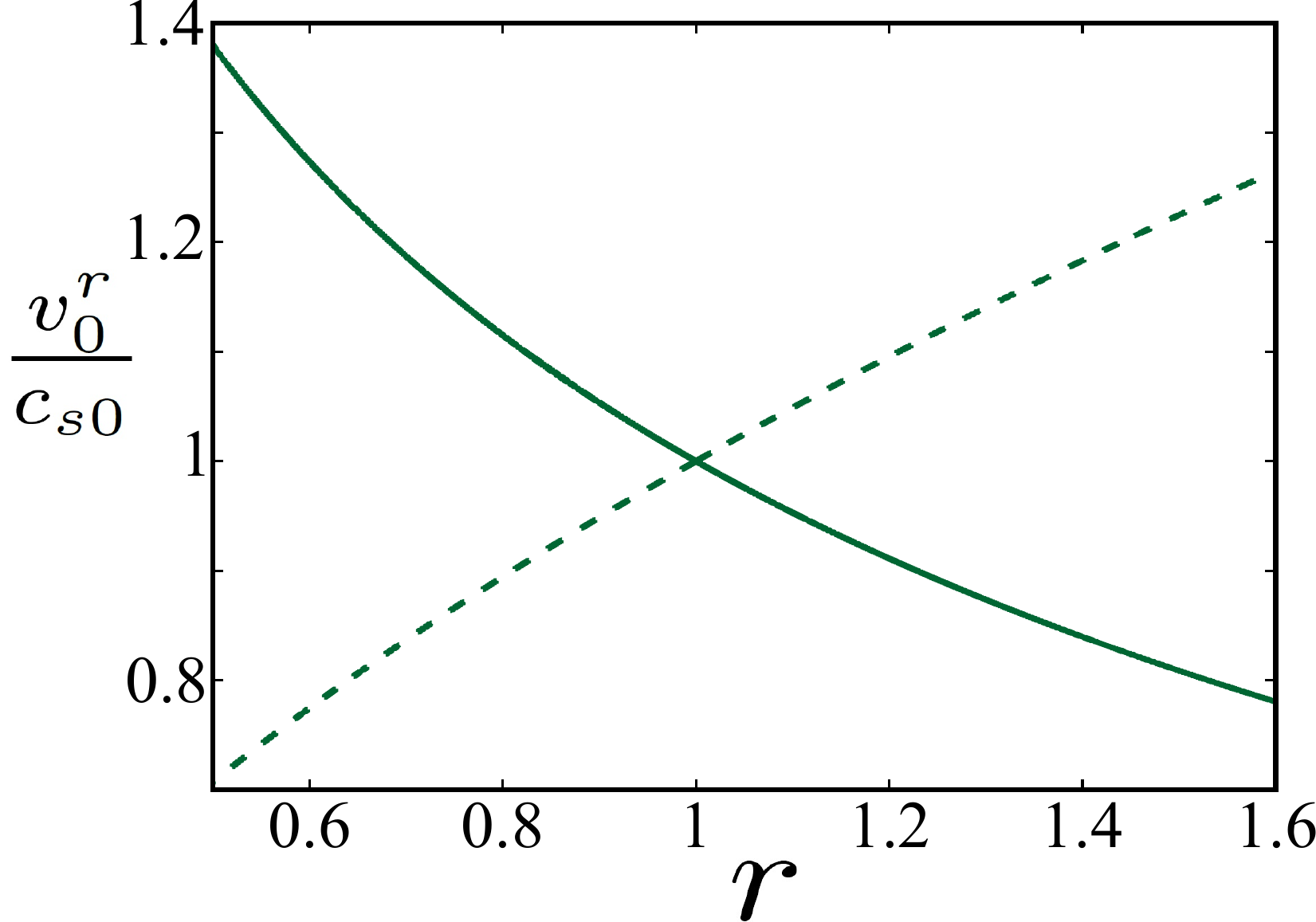}
\caption{{\it Steady radial flow transonic solution.} Solid green lines correspond to an ABH for radial inflow.
Mach number vs $r$; with a fixed $C_1$ and $C_2$, we get two transonic solutions. For a 
BEC with $V_{\rm ext}(r)=-\frac{1}{r}$ and $C_2=0.5$, we have $r_H=1$, $c_{sH}=1$}
\label{Mar}
\end{figure}

\subsection{Null   geodesics traveling away from the trapped surface}
The orthogonal direction to a $r=$ constant hypersurface is $\pm\partial_\mu r=(0,\pm 1, 0)$ according to Eq. \eqref{tangent}; since the spacetime is stationary, this radially outward or inward direction to the $r=$ constant hypersurface is time independent. 
Null geodesics satisfy $ds^2=0$, giving a Galilean velocity addition rule
\begin{equation}\label{soundray}
c_{s0}{\bm e}_n=\frac{d\vec{ x}}{dt} -\left(-v_{0}{\bm e}_r+\frac{l}{r}{\bm e}_\phi\right),
\end{equation}
$v_0^r=-v_0$, ${\bm e}_n$ is a unit direction, specifying the net direction of a null geodesic $\frac{d\vec{x}}{dt}$, and can be any unit vector in the tangent space on the $x-y$ plane at constant $t$; evidently, different null geodesics have different ${\bm e}_n$s. We rewrite Eq.~\eqref{soundray} as follows 
\begin{equation}\label{nulld}
\frac{d\vec{ x}}{dt}=\left(c_{s0}({\bm e}_n\cdot{\bm e}_r)-v_{0}\right){\bm e}_r+\left(c_{s0}({\bm e}_n\cdot{\bm e}_\phi)+\frac{l}{r}\right){\bm e}_\phi.
\end{equation}
Here, $\cdot$ is the dot product of two vectors in Euclidean space.
We observe from Eq.~\eqref{nulld} that at a fixed $r$, null geodesics with ${\bm e}_n=\partial_\mu r={\bm e}_r$ have the maximum $\frac{d\vec{ x}}{dt}$ along the radially outward ($+{\bm e}_r$) direction. This is the direction along which null geodesics diverge maximally if the $r=$ constant hypersurface is outside the event horizon.  Since inside the event horizon, $c_{s0}<v_{0}$, the outward pointing null geodesic orthogonal to $T^1$ also converges (radial component of $\frac{d\vec{x}}{dt}$ is negative) initially at $T^1$. Therefore, any $r=$ constant hypersurface at constant time slice in the $\frac{v_{0}}{c_{s0}}>1$ region is a one-dimensional trapped surface in the analogue spacetime, representing the generalization of Penrose's definition of two-dimensional trapped surfaces inside a $3+1$D black hole horizon.

For $l=0$, null geodesics along $\pm{\bm e}_r$ satisfy
\begin{equation}\label{rio}
\frac{dr}{dt}=-v_{0}\pm c_{s0}.
\end{equation}
Since for pure radial flow, no velocity is added into the null geodesics in the ${\bm e}_\phi$ direction, null geodesics with ${\bm e}_n\cdot{\bm e}_\phi=0$ at $t=0$ satisfy Eq.~\eqref{rio} for all time, i.e., ${\bm e}_n\cdot{\bm e}_\phi$ remain zero $\forall t$. 
If we assume flows corresponding to the analogue black hole solution exist in the domain $M$, to reach $r\rightarrow 0+$ from a trapped surface of radius $r_T$ requires the time 
\begin{equation}\label{reachingS}
\Delta t_{\pm}=\lim_{r_0\to 0+}\int_{r_T}^{r_o} dr \frac{1}{-v_{0}\pm c_{s0}}.
\end{equation}
This integral is finite for physical flows. We also note that $\Delta t_+>\Delta t _-$.

We now calculate the corresponding 
affine parameter interval for radial rays for the ABH with $l=0$.  We first specify the null geodesic curves by the time variable $t$. Not knowing whether laboratory time $t$ is an affine parameter or not, the geodesic equation is in general given by \cite{carroll_2019}
\begin{equation}\label{GE}
\frac{d^2 x^\mu}{dt ^2}+\Gamma ^{\mu}_{\nu\sigma}\frac{dx^\nu}{dt}\frac{dx^\sigma}{dt}=f(t)\frac{dx^\mu}{dt},
\end{equation}
where $+\Gamma ^{\mu}_{\nu\sigma}$ is the Christoffel symbol.
If $t$ happens to be an affine parameter then $f(t)$ would be zero at any $t$. In our $2+1$D  system, $x^\mu \equiv (t,r,\phi)$. For radial null geodesics $\frac{d\phi}{dt}=0$. Therefore, putting $x^\mu =t$ in Eq.~\eqref{GE}, we get
\begin{equation}
f(t)=\Gamma ^t _{tt}+2\Gamma ^{t}_{tr} \frac{dr}{dt}+\Gamma ^t_{rr} \left(\frac{dr}{dt}\right)^2.
\end{equation}
The affine parameter $\lambda $ is related to $t$ by
\begin{equation}
\frac{d^2\lambda}{dt^2}=f(t)\frac{d\lambda}{dt}.
\end{equation}
Therefore, in functional form we have 
\begin{equation}
\lambda =\int \exp \left(\int f(t) dt \right) dt.
\end{equation}
We now make a change of variable $t\rightarrow r$,
\begin{equation}\label{affinet}
\int f(t) dt =\int \frac{F(r)}{(dr/dt)} dr.
\end{equation}
$\frac{dr}{dt}$ is a function of $r$ for ingoing and outgoing null geodesics, cf.~Eq. \eqref{rio}, and $F(r)$ is the functional form of $f(t)$ under the change in variable $t\rightarrow r$.
For the metric \eqref{agmn},
\begin{eqnarray}
\Gamma ^{t}_{tt}&=&-\frac{1}{2}g^{tr}\frac{dg_{tt}}{dr},\\
\Gamma ^{t}_{tr}&=&\frac{1}{2}g^{tt}\frac{dg_{tt}}{dr},\\
\Gamma ^{t}_{rr}&=&g^{tt}\frac{d g_{tr}}{dr}+\frac 12 g^{tr}\frac{dg_{rr}}{dr}.
\end{eqnarray}
with $v_0^r=-v_0$ in the metric \eqref{agmn}, we find for an ingoing ray 
\begin{equation}
f(t)\equiv F(r)=g^{tt}\Omega ^2\left[c_{s0}^2\frac{d(c_{s0}+v_{0})}{dr}+c_{s0}(c_{s0}+v_{0})\frac{dc_{s0}}{dr}\right]+g^{tt}c_{s0}^2\left(c_{s0}+v_{0}\right)\frac{d\Omega ^2}{dr},
\end{equation}
where $\Omega ^2$ is the conformal factor $\left(\frac
{\rho_0}{c_{s0}}\right)^2$ in the metric \eqref{agmn}. 
Then we find from Eq.~\eqref{affinet} the functional form of $\lambda =\lambda (r)$
\begin{equation}\label{lin}
\lambda = -\int \Omega ^2 c_{s0} dr.
\end{equation}
Similarly, for an outgoing ray,
\begin{equation}\label{lout}
\lambda = \int \Omega ^2 c_{s0} dr.
\end{equation}
The overall sign is immaterial, because an affine parameter remains an affine parameter under an affine transformation $\lambda \rightarrow (a\lambda +b)$, where $a$ and $b$ are constants.

\subsection{Null expansion for radial null geodesics inside trapped surface}
The null expansion $\Theta$ for a null geodesic $l^\mu$ in affine parametrization is defined as
\begin{equation}
\Theta =l^\mu_{~;\mu},
\end{equation}	
where $l^\mu=\frac{dx^\mu}{d\lambda}$, $\frac{dx^\mu}{d\lambda}$ is tangent vector for the null geodesic and 
$\lambda$ is affine parameter. A radial null geodesic, parametrized by $t$, is given by
\begin{equation}
\tilde{l}^\mu_{\pm} = (1, -v_0 \pm c_{s0},0),
\end{equation}
where $\pm$ sign specifies radially out/inward pointing null geodesic congruences, respectively.
We have $d\lambda=-c_{s0}^3 dr$ (see main text), and we conventionally employ the minus sign 
because in our ABH, inside the black hole horizon, both radially out- and inward pointing null geodesics move radially inside. We choose $\lambda $ to increase as $r$ decreases for these radial rays.
We find that a null vector inside black horizon in affine parametrization reads 
\begin{equation}
l^\mu_{\pm} = \mp\frac{1}{c_{s0}^3(c_{s0} \mp v_0)}\tilde{l}^\mu_{\pm}=\mp\frac{1}{c_{s0}^3(c_{s0} \mp v_0)}(1, -v_0\pm c_{s0},0).
\end{equation}
In the ABH spacetime \eqref{agmn} with $l=0$, $v_0 ^r =-v_0$, 
$\Theta $ is identical for radially outward as well as for radially inward pointing null geodesics, and is given by
\begin{equation}
\Theta = \frac{1}{2}\left(-\mathscr{M}+\frac{r}{c_{s0}}\frac{dv_0}{dr}\right).
\end{equation}
Here, $\mathscr{M}$ is Mach number, i.e., $\frac{v_0}{c_{s0}}$.
Using the expression of $\frac{dv_0}{dr}$ in Eq.~\eqref{dvdr}, we find 
\begin{equation}
\Theta = \frac{1}{2}\left(-\mathscr{M}+\frac{\frac{1}{\mathscr{M}}-\frac{1}{rv_0c_{s0}}}{\left(1-\frac{1}{\mathscr{M}^2}\right)}\right).
\end{equation}
We know that $c_{s0}<1/r^{1/3}$ in the supersonic flow region. Using Eq.~\eqref{csv} for $\gamma =2$, we get 
\begin{equation}
\mathscr{M}=\frac{1}{c_{s0}^3 r}.
\end{equation}
Therefore, we have 
\begin{equation}
\lim _{r\to 0}\mathscr{M}=\infty 
\end{equation}
and 
\begin{equation}
\lim _{r\to 0}\Theta = -\infty.
\end{equation}

\subsection{Singularity theorem in non-Einsteinian gravity}
The ABH in the main text arises from an axisymmetric flow, we now proceed to a  singularity theorem for a general spacetime.

{\bf Theorem 4.} {If a spacetime with a noncompact Cauchy hypersurface contains a trapped surface, and if the Ricci tensor contraction $R_{\mu\nu}l^\mu l^\nu \geq 0$ for any null vector $l^\mu$, then the spacetime is null geodesically incomplete.}
\begin{proof}
The proof proceeds similar the one given by Witten in his review \cite{Witten} except that here we do not  have Einstein equations determining the metric. 
The Raychaudhuri equation, for a  null geodesic congruence in a spacetime of dimension $D$ ($>2$) is given by \cite{Witten,Kar2007},
\begin{equation}\label{LRC}
\dot{\Theta}+\frac{\Theta ^2}{(D-2)}=\omega ^2 -\sigma ^2-R_{\mu\nu}l^\mu l^\nu,
\end{equation}
where the semicolon $_{;}$ represents covariant derivative, $l_\mu l^\mu =0$, $\Theta = l^\mu_{~;\mu}$, $\sigma _{\mu\nu}=\frac{1}{2}\left(l_{\mu_;\nu}+l_{\nu;\mu}\right)-\frac{1}{(D-2)}h_{\mu\nu}\Theta$, $h_{\mu\nu}=(g_{\mu\nu}+l_\mu N_\nu + l_\nu N_\mu)$, $N_\mu N^\mu =0$, $N^\mu l_\mu =-1$, $\omega _{\mu\nu}=\frac{1}{2}\left(l_{\mu_;\nu}-l_{\nu;\mu}\right)$, $\sigma ^2 =\sigma ^{\mu\nu}\sigma _{\mu \nu}=\sigma ^\mu _{~\nu} \sigma ^\nu _{~\mu} \geq 0$, $\omega ^2 =\omega ^\mu _{~\nu} \omega ^\nu _{~\mu} \geq 0$. $\dot{\Theta}=\frac{d\Theta}{d\lambda}$ where $\lambda$ is an affine parameter parametrizing a null geodesic.

We notice if one were to use the Einstein equations $R_{\mu\nu}l^\mu l^\nu=8\pi T_{\mu\nu}l^\mu l^\nu$, in units with $G=c= 1$. Therefore, the null energy condition $T_{\mu\nu}l^\mu l^\nu\geq 0$ is equivalent to $R_{\mu\nu}l^\mu l^\nu \geq 0$ when the Einstein equations apply. 

It is possible to choose a coordinate system such that $\omega _{\mu\nu}$ vanishes \cite{Witten}; therefore, in that coordinate system, the sign of the right-hand side of 
Eq.~\eqref{LRC} is certainly negative if $R_{\mu\nu}l^\mu l^\nu \geq 0$. All the next steps of the proof 
have been given by Witten in \cite{Witten}. The null geodesics truncate but {\it do not intersect} when $\Theta$ goes to minus infinity (defining a focal point as a point where $\Theta$ is minus infinity \cite{Witten}), at a finite affine parameter value $\frac{(D-2)}{\Theta _0}$. Here, $\Theta =-\Theta _0$, with $\Theta _0 >0$ on the trapped surface. where one starts with zero affine parameter value. The  existence of a trapped surface ensures that the null expansion $\Theta $ is negative on the trapped surface. 
In our case of $D=3$, the affine parameter value would then be bounded by $\frac{1}{\Theta _0}$.
\end{proof}
Our singularity theorem  does not employ the Einstein equations 
as in \cite{Penrose65PRL,Hawking0, Hawking1, Hawking2, Hawking3,Hawking1970singularities}. 
For sonic spacetimes or other non-Einsteinian theories of gravity,
we can not relate the null convergence factor $R_{\mu\nu} l^\mu l^\nu$ to the null energy condition: 
In fluid-dynamical analogue gravity, while an Einstein tensor
can as usual be constructed from the metric, an energy-momentum tensor
 $T_{\mu\nu}$  appearing on the right-hand side of would-be analogue Einstein equations can not be identified.  
We therefore restrict ourselves to employing the null convergence condition $R_{\mu\nu}l^\mu l^\nu  \geq 0$ 
instead of the null energy condition $T_{\mu\nu}l^\mu n^\nu \geq 0$.  
Then a singularity theorem deriving from the Raychaudhuri equation \cite{RC,Hawking_Ellis_1973,dadhichderivation,Kar2007}, 
 that is purely from the geometry of the spacetime, 
without invoking the Einstein equations, can be formulated.

However, the absence of the null energy condition makes the theorem less powerful. In the case of Penrose singularity problem in Einstein gravity, given a spacetime, we need to identify a trapped surface, and then we need to declare $T_{\mu\nu}$ satisfies a physical condition, i.e., the null energy condition (or the more restrictive weak energy condition), and that would guarantee null geodesic incompleteness. There is of course in principle the possibility of exotic matter so that Penrose's singularity is avoided.  However, in our singularity theorem, we need to calculate $R_{\mu\nu}l^\mu l^\nu$ for all null geodesics or at least for a {\it suitable} null geodesic congruence to check whether the congruence encounters a caustic within a finite affine parameter interval from the trapped surface. The Raychaudhuri equation for null geodesic congruence guarantees if  $R_{\mu\nu}l^\mu l^\nu\geq 0$, the null geodesic congruence would encounter a caustic 
within a finite affine parameter interval.  If the spacetime is symmetric enough, from this symmetry we can choose suitable null geodesic congruences, and then we calculate the affine parameter for those null geodesic congruences. The work \cite{visseraffine} describes the computation of the affine parameter in spacetime with symmetries. After finding the affine parameter, one can calculate the null expansion $\Theta $ for those chosen null geodesic congruences instead of calculating $R_{\mu\nu}l^\mu l^\nu$, as we did in our ABH example. For a null geodesic congruence, whether computing $\Theta $ instead of computing $R_{\mu\nu}l^\mu l^\nu$ would be computationally easier  depends on the given spacetime and the chosen null geodesic congruence.  We must also note that the Penrose singularity theorem as well as our singularity theorem for non-Einsteinian gravity provide necessary, but not sufficient conditions for null geodesic incompleteness at the caustic within a finite affine interval  into the future, starting from a trapped surface, i.e., for a 
Penrose-type singularity. There is always the possibility that $R_{\mu\nu}l^\mu l^\nu <0$ but $R_{\mu\nu}l^\mu l^\nu +\sigma ^2 \geq 0$, and this possibility would also lead to a Penrose-type singularity. Indeed, $R_{\mu\nu}l^\mu l^\nu +\sigma ^2 \geq 0$ is not considered as any singularity theorems' condition \cite{Wald} because $R_{\mu\nu}l^\mu l^\nu +\sigma ^2 $ requires computation of $\sigma _{\mu\nu}$, and $\sigma_{\mu\nu}$ contains $\Theta $. If $\Theta $ is computed, there is no need to compute $R_{\mu\nu}l^\mu l^\nu +\sigma ^2$, because $\Theta$ in affine particularization will directly provide the information about caustic points in a given spacetime. We observe that the axisymmetric ABH that we considered in the main text actually does not satisfy the null convergence condition of Theorem 4 for radial null geodesic congruences; however, still the radial null geodesics from a trapped surface converge at $r=0$ in a finite affine interval, as we concluded by directly calculating $\Theta$ itself.

\end{widetext}

\end{document}